\documentclass[11pt]{llncs}
\usepackage[T1]{fontenc}
\usepackage[utf8]{inputenc}

\usepackage{amsmath}
\allowdisplaybreaks
\usepackage{amssymb}
\usepackage{amsthm}
\usepackage{graphicx}
\usepackage{mathtools,dsfont}
\usepackage{bm}
\usepackage{bbm}
\usepackage{enumerate}

\def\withcolors{1}
\def\withnotes{1}

\ifnum\withcolors=1
  \newcommand{\ch}[1]{{\color{blue}{#1}}}

\else

\fi

\usepackage[colorinlistoftodos,textsize=scriptsize]{todonotes}
\ifnum\withnotes=1

\else

\fi

\def\I{{\ell}}
\def\J{{h}}
\def\bTheta{{\bf{\Theta}}}
\def\bSigma{{\bf{\Sigma}}}
\def\mF{{\mathcal{F}}}

\def\thetablank{{}}

\newcommand{\typical}[1]{{A_{\varepsilon}^{(n,{#1})}}} 
\newcommand{\bartypical}[1]{{\bar{A}_{\varepsilon}^{(n,{#1})}}} 
\newcommand{\sig}[1]{\sigma_{#1}} 
\newcommand{\swapsig}[1]{\sigma^{#1}} 
\newcommand{\indexswapsig}[2]{\sigma_{#1}^{#2}} 
\newcommand{\indexset}[2]{\mathcal{I}_{{#1}^{#2}}} 
\newcommand{\swappedind}[1]{i_{#1}} 
\newcommand{\numbadteams}[1]{m_{00}\left({#1}\right)} 

\usepackage{lmodern}
\usepackage{xspace}
\usepackage[protrusion=true,expansion=true]{microtype}
\usepackage{fullpage}

\usepackage{url}

\usepackage{multirow}

\usepackage{dsfont} 
\usepackage{fullpage}

\makeatletter
\@ifundefined{theorem}{
  \theoremstyle{definition}
  \newtheorem{definition}{Definition}
  \theoremstyle{plain}
  \newtheorem{theorem}{Theorem}
  \newtheorem{corollary}{Corollary}
  \newtheorem{lemma}{Lemma}

  \theoremstyle{remark}

  \newtheorem{example}{Example}

}{}
\makeatother
\newenvironment{proofof}[1]{\noindent{\bf Proof of {#1}:~~}}{\(\qed\)}
 \newtheorem{prop}{Proposition}

\newcommand{\ignore}[1]{}

\newcommand{\EE}{\mathbb{E}}

\newcommand{\PP}{\mathbb{P}}

\newcommand{\expectation}[1]{\EE\left[#1\right]}

\def\ignore#1{}

\def\orpro{\mathop{\mathchoice
   {\vee\kern-.49em\raise.7ex\hbox{$\cdot$}\kern.4em}
   {\vee\kern-.45em\raise.63ex\hbox{$\cdot$}\kern.2em}
   {\vee\kern-.4em\raise.3ex\hbox{$\cdot$}\kern.1em}
   {\vee\kern-.35em\raise2.2ex\hbox{$\cdot$}\kern.1em}}\limits}

\def\andpro{\mathop{\mathchoice
 {\wedge\kern-.46em\lower.69ex\hbox{$\cdot$}\kern.3em}
 {\wedge\kern-.46em\lower.58ex\hbox{$\cdot$}\kern.25em}
 {\wedge\kern-.38em\lower.5ex\hbox{$\cdot$}\kern.1em}
 {\wedge\kern-.3em\lower.5ex\hbox{$\cdot$}\kern.1em}}\limits}

\def\simge{\mathrel{%
   \rlap{\raise 0.511ex \hbox{$>$}}{\lower 0.511ex \hbox{$\sim$}}}}

\def\simle{\mathrel{
   \rlap{\raise 0.511ex \hbox{$<$}}{\lower 0.511ex \hbox{$\sim$}}}}

\newcommand{\EEc}[2]{\EE\left[#1\mid#2\right]} 

\title{Information Signal Design for Incentivizing Team Formation\thanks{A preliminary version of this work was presented at WINE 2018, and accompanied by a one-page extended abstract \cite{hssaine2018information}.}}
\author{Chamsi Hssaine, Siddhartha Banerjee}
\institute{School of Operations Research and Information Engineering, Cornell University\\
\email{\{ch822,sbanerjee\}@cornell.edu}}

\date{}
\begin{document}

\maketitle

We study the use of Bayesian persuasion (i.e., strategic use of information disclosure/signaling) in influencing endogenous team formation.
This is an important consideration in settings such as crowdsourcing competitions, open science challenges and group-based assignments, where a large number of agents self-organize into small teams {which then compete against each other.} A central tension here is between the strategic interests of agents who want to join the highest-performing team, and that of the principal who wants to maximize social welfare. Moreover, although the principal cannot choose the teams or modify rewards, she often has additional knowledge of agents' abilities, and can leverage this information asymmetry to provide signals that influence team formation. 
Our work uncovers the critical role of self-awareness (i.e., knowledge of one's own abilities) for the design of such mechanisms. 
In settings with binary agent abilities, when agents are agnostic of their own abilities, we provide signaling mechanisms which incentivize teams achieving the first-best welfare under convex utility functions, and mechanisms which are asymptotically efficient for concave utility functions.
On the other hand, when agents are self-aware, then we show that no signaling mechanism can do better than not releasing information, while satisfying agent participation constraints. 

\section{Introduction}

Consider a course instructor who wants to group students into teams for assignments. 
Each student has an unknown `aptitude' for the course, with past performance (academic background, GPA, etc.) providing a prior on these aptitudes; moreover, everyone prefers teammates with higher aptitudes. In the absence of any other information, the priors translate into a perceived ranking over students, and if left alone, students will team up according to this ranking (e.g., for teams of two, the top two form one team, the next two form another, and so on). On the other hand, the instructor would prefer it if high-performing and low-performing students worked together, to encourage better learning outcomes. However, she does not want to decide the teams herself, or alter the grading to incentivize such teams to form. Is there anything she can do in such a setting?

The above problem takes on a more meaningful form in the context of collaboration in open science challenges and crowdsourcing competitions~\cite{bender2016crowdsourced}.
As an example, consider the DREAM Challenges~\cite{dreamchallenge} -- an online crowdsolving platform which leverages researchers from various backgrounds to solve problems in biology and medicine. A critical design feature of these challenges is that after an initial exploration round, competitors are required to pair up into teams -- this is inspired by observations that the performance of individual contestants' solutions perform much worse than ensembles of these solutions~\cite{marbach2009combining}.
The team composition however is completely decided by the participants, with minimal involvement of the designers (and no change in reward structures). Thus, gaining insight into team formation and dynamics is key to the success of these platforms; this sentiment is echoed in a report from the National Academy of Sciences highlighting the need to find ways to foster team effectiveness~\cite{national2015enhancing}.

{Algorithmic solutions to the problem of designing effective teams are currently being used in institutions globally; the Comprehensive Assessment of Team Member Effectiveness (CATME) software, used in close to 2,000 universities across the United States, is one such example~\cite{catme}. 
This software surveys students about desirable criteria determined by the instructor and heuristically determines team assignments based on these criteria~\cite{layton_design_2010}.}
	
{Though such a system facilitates the assignment of students to teams -- particularly in large settings -- one drawback is the lack of \emph{organic} teams. Self-selected teams have been linked to increased group cohesiveness, accountability, and cooperativeness, and thus an overall improved team experience~\cite{layton_design_2010}.} {This explains the surge of interest in the organic formation of teams within the context of Massive Open Online Courses (MOOCs). NovoEd serves as a prime example of a platform devoting its efforts to team formation, viewing collaboration and team-based learning as a means to solve the problem of high attrition rates from which many MOOCs suffer~\cite{novoedtech}. In order to create an engaging learning experience, NovoEd ran the following experiment. After first signing up for a course, students were assigned to small groups (based on interest, location, background, etc.). Throughout the course, students within the same group would privately rate each other, and by the end of the course a student ranking was compiled. This ranking then allowed students to decide how to form groups on their own for subsequent courses. The result of such an experiment was evident. Compared to a retention rate of 10\% for the average MOOC, NovoEd reported a retention rate of 50-70\% for students who completed the first assignment~\cite{stanfordnovoed}.} 

{Motivated by these real-world examples,} our work focuses on the use of \emph{Bayesian persuasion} (i.e., strategic information revelation) for incentivizing team formation. The main idea is that many strategic settings have an inherent information asymmetry, wherein the principal has more information than participating agents. By controlling the release of this information, the principal can influence agents’ decisions. For example, consider the formation of assignment groups in MOOCs: in such settings, students often have (common) priors about each others' abilities (for example, based on their educational background), while the instructor may have more accurate estimates of student abilities based on scores from early tests. How she releases the results of this test can affect what teams are formed. If she chooses not to release any information, then the students tend to team up according to their perceived aptitudes (i.e., their priors about each others' abilities); on the other hand, if she releases the scores as is, the students may team up according to the test scores. The main idea we pursue in this work is to understand if there is any way of releasing the scores that can lead to {socially optimal teams}; in other words, we want to know how the principal can \emph{influence endogenous team formation using strategic signaling}.

\subsection{Overview of Model and Results}

We consider a setting with $n$ agents who endogenously form teams, leading to some utility for each agent. The teams are chosen endogenously by the agents, in the form of a \emph{stable matching} (in case of two-member teams; more generally, in a group-\ch{stable} manner). The principal however can influence agents' preferences via strategic release of information.

In more detail: each agent has an intrinsic type, drawn from some public prior. {We consider two scenarios: where agents are self-agnostic (i.e., uncertain of their own type), and where they are self-aware (i.e., know their own type).} Crucially, we always assume that every agent's type is known to the principal, and unknown to other agents. Each agent's utility is an increasing function of her type and her teammates' types; the principal's goal is to maximize {social welfare.} 
Henceforth in this work, we focus on settings with $n$ agents with $K$ different priors, and binary type-space $\{0,1\}$.

The main tool available to the principal is \emph{Bayesian persuasion}, whereby she can leverage her information asymmetry by \emph{committing to a signaling scheme based on the {realized} types}. This signaling scheme can be verified by the agents (for example, the principal can commit to using an open-source script that inputs the true types and generates the signal). Thus, the signal affects the agents' posterior over the types, which then determines their choice of teams via a stable matching. The aim of the principal is to choose a signaling scheme which is Pareto improving {over the no-information outcome} (which is a natural endogenous participation constraint, requiring that all agents are weakly better-off by agreeing to receive the signal), and for which the resulting stable matching maximizes the social welfare.

In the context of the above setting, our contributions are summarized as follows:
\begin{enumerate}
	\item We characterize the optimal signaling schemes in the form of a linear program which implements a \emph{persuasive recommendation} -- a consistent posterior ranking of the agents. This LP however has $\Theta(2^nn!)$ variables, and hence is computationally intractable.
	\item In settings where agents are \emph{self-agnostic} (i.e., are uncertain about their own type), we demonstrate the following:
	\begin{enumerate}[(i)]
	\item for \emph{convex} utility functions, full information revelation is the optimal signaling scheme;
	\item for \emph{strictly concave} utility functions
	\begin{enumerate}[a.]
		\item under {uniform priors} (i.e., $K=1$ clusters, wherein agents have i.i.d. types), we present a signaling mechanism which realizes the \emph{first-best} matching (and hence is optimal);
	\item for finite $K$, we present a signaling mechanism that is asymptotically optimal in $n$.
	\end{enumerate}
	\end{enumerate}
	\item In contrast, in settings where agents are \emph{self-aware} (i.e., know their own type), for both convex and concave utility functions we demonstrate a strong \emph{impossibility result} — we show that even under a uniform prior ($K=1$), no signaling scheme can do better than not releasing any information, while satisfying {Pareto improvement} constraints. Moreover, this strong impossibility is tied to a ``blocking set'' of agents which depends on the concavity of the utility function.
\end{enumerate}
For ease of exposition, we first demonstrate the above results in the context of teams of size two. Later, we show how our results extend to teams of arbitrary (constant) size.
Moreover, though the above results are for binary type-spaces, our techniques are based on underlying symmetries and measure concentration properties which are not particular to these assumptions, and thus should generalize to more complex settings. We discuss this in more detail in Section~\ref{sec:extensions}. 

Overall, our work provides important insights and techniques for the design of Bayesian persuasion schemes for more general team formation settings. In particular, our results indicate the importance of self-awareness in determining the success of signaling mechanisms, and provide a tractable policy for self-agnostic settings based on intra-cluster pairing of type-profiles. Moreover, showing this strategy is asymptotically optimal requires a novel dual-certification argument, which may be useful in related settings. Finally, our results highlight the tension between two natural desiderata for the problem of team formation: \emph{stability} and \emph{endogenous participation}. Even in simple settings, this tension is apparent and presents significant challenges. Importantly, our work formalizes the fact that each of these desiderata is tied to a particular `blocking set' of agents, depending on the setting: when agents are self-aware, stability is linked to satisfying incentives of agents with stronger priors; when agents are self-agnostic, endogenous participation is critical for stronger agents in the case of concave utilities, and for weaker agents in the case of convex utilities. This corresponds to intuitive beliefs one may have about incentives and the role of information in these settings. In a sense, our work provides formal micro-foundations for this intuition.

\subsection{Related work}

{The operations research and computer science communities have produced an extensive line of work on the team formation problem, where the objective is to form teams of individuals with all the required skills to perform a certain set of tasks. Variants of the problem that have been studied include measuring communication overhead to create effective teams in networks~\cite{lappas_finding_2009}, adding workload constraints on users~\cite{majumder_capacitated_2012}, and forming teams in an {online} fashion, as tasks arrive dynamically~\cite{anagnostopoulos_online_2012}. This line of work, however, ignores agent \emph{incentives}, i.e., the idea that agents have preferences over each other, and act strategically based on these.

On the other hand, there is a large body of work that considers the role of strategic decisions in the formation of various combinatorial structures among agents. Two exemplars are the study of \emph{coalitional games}, which considers equilibrium models for the formation of coalitions among agents~\cite{ray_theory_1999,bogomolnaia_stability_2002}, and \emph{network formation games}~\cite{jackson_formation_2002,galeotti2006network}, which studies the formation of social networks, with various costs and benefits for creating links between players.} This literature on coalitional and network games, however, focuses primarily on equilibrium characterizations, rather than considering the design of mechanisms towards maximizing some central objective. In particular, we note that there does not seem to have been work on signaling mechanisms to incentivize the formation of specific teams, our ultimate objective.

{Our focus on the use of strategic signaling as an incentive mechanism places our work squarely} within the framework of \emph{Bayesian persuasion}, a topic which has garnered much recent attention. We briefly summarize some of the main ideas of this topic below; for a more detailed survey of this literature, cf.~\cite{dughmi2017survey}.

The basic idea originates from the seminal work of Kamenica and Gentzkow~\cite{kamenica2011bayesian}, which considers a principal who commits to a signaling policy which maps the true state of the world to a signal sent to a single agent, and derives conditions on the principal's utility function under which she strictly benefits from persuasion. Going beyond, Kremer et al. \cite{kremer2014implementing} consider a dynamic setting in which agents arrive sequentially and choose an action with an unknown (but deterministic) reward. The goal of the principal is to find the optimal disclosure (or recommendation) policy of a planner who wants to maximize social welfare, whereas agents simply seek to maximize their own expected reward. The authors show that the optimal policy is a threshold policy which explores as much as possible, and then always recommends the best action, and obtain a similar near-optimal threshold policy with stochastic rewards. Recent papers greatly generalize this line of work~\cite{mansour2015bayesian,bahar2015economic}. 

In the context of multi-agent settings with no externalities, Arieli and Babichenko~\cite{arieli2016private} look at a model where the principal tries to persuade individuals to adopt a product by sending \emph{private} signals. They characterize the optimal policy for supermodular, submodular, and supermajority utility functions of the principal. More recent work has been devoted to proving hardness or inefficacy results for finding the optimal information disclosure policy in such settings~\cite{dughmi2014hardness,dughmi2017algorithmic}. 

In contrast to these papers, our work involves a multi-agent scenario with \emph{externalities} -- in other words, not only is the principal playing a game with multiple agents, but the agents are playing a game amongst themselves. Recent work has looked at related models in the context of routing and queueing games. In \cite{bhaskar2016hardness}, the authors present hardness results on \emph{public} signaling for Bayesian two-player zero-sum games and Bayesian network routing games; more recent works~\cite{das2017reducing,tavafoghi2017informational} consider practical variants of such policies in restricted settings. 
In the setting of strategic delay announcements for queueing,Iyer and Lingenbrink show that when the principal is a revenue maximizer, a binary signaling mechanism with a threshold structure is optimal~\cite{lingenbrink2017optimal}. {Similar insights hold in an entirely different problem -- that of optimal signaling of content accuracy on online social networks~\cite{candogan2017optimal}. The authors  show that, when users derive positive utility from their neighbors engaging with content on the platform, the optimal mechanism also has a simple threshold structure.} Our work however is, to best of our knowledge, the first to consider the problem of finding the optimal signaling policy for team formation.

\section{Preliminaries}

\subsection{Basic Setup}

We consider a setting with a principal and $n$ agents ($n$ even), where the agents endogenously partition themselves into two-member teams. {We will show in Section~\ref{ssec:larger-size} how to extend the model to a setting with teams of size $a > 2$.}

Each agent $i \in [n]$ has a random type $\theta_i \in \{0,1\}$, which can be interpreted as her intrinsic ability to perform the task at hand. Agents are exogenously clustered into $K$ groups based on having common priors, with each cluster $k$ composed of $n_k$ agents. (For notational convenience, we assume throughout that $n_k$ is even.) In the spirit of the framework used in other work on optimal signaling~\cite{kamenica2011bayesian,dughmi2017algorithmic,arieli2016private}, we assume that, if agent $i$ is in cluster $k$, the type of agent $i$ is drawn independently from a $Ber(p_k)$ distribution. We assume that $p_k \neq p_{k'}$ for all $k \neq k'$.  Importantly, this implies that agents in the same cluster are \emph{homogeneous}. We use $\bTheta \triangleq \{0,1\}^{n}$ to denote the space of type-profiles of all $n$ agents, and denote $\lambda$ to be the product distribution over $\bTheta$. Note that this notion of clusters is natural in the settings we are interested in; for example, in the context of crowdsourcing, these represent groups of agents whose abilities have been estimated to be very similar.

For any given type-profile $\theta\in\bTheta$, we denote $\I_k(\theta), \J_k(\theta)$ the number of agents with a true type of 0 and 1, respectively, in cluster $k$.  Similarly, we define $\J(\theta)=\sum_i\mathds{1}\{\theta_i=1\}$ to be the total number of agents with type $1$ (i.e., `high' type) in $\theta$, and $\I(\theta)=n-\J(\theta)$ to be the total number of agents with type $0$ (i.e., `low' type) in $\theta$.  


Each agent's payoff depends only on her own type and that of her match. We denote the utility function of agents as $u: \{0,1\}\times\{0,1\}\to\mathbb{R}$ and assume it obeys the following properties:
\begin{enumerate}[nosep]
	\item[i.] symmetric (i.e., $u(\theta_i,\theta_j) = u(\theta_j,\theta_i)$)
	\item[ii.] strictly increasing in each type (i.e., $u(\theta,0) <  u(\theta,1)$).
\end{enumerate}
{This utility function represents the quality of the team. One can also think of it as the likelihood that the team will successfully complete the task. Note that this function is common \emph{across all agents}, and can thus be thought of as an objective, rather than perceived, metric of success. Additionally, implicit in these properties is the assumption that agents within the same team benefit equally from the match. Such a model has been used in data-driven studies of the team formation problem~\cite{eftekhar_team_2015}, and applies to a variety of settings, including: competitions where all team members receive the same award or recognition; online courses, where everyone receives the same grade for an assignment; public works, where an entire community benefits from the completion of an infrastructure project.}


\subsubsection{Information structure.} As is typical in the Bayesian persuasion framework, we assume that the distribution from which each agent is drawn is known to all agents, but that \emph{each agent's realized type is unknown to other agents}. On the other hand, we assume that \emph{the principal has full knowledge of the realized types} $\theta = (\theta_1,\ldots,\theta_n)$.
Finally, with regard to what an agent knows about her own type, we consider two cases: one in which agents are \emph{self-aware}, and one in which agents are \emph{self-agnostic}. 
\begin{definition}
A \emph{self-aware} agent is an agent who knows her own type. A \emph{self-agnostic} agent is an agent who does \emph{not} know her own type.
\end{definition}

For example, a self-agnostic student in a course is one who has no prior experience in the subject, and hence no idea of her aptitude for it; on the other hand, a self-aware student has some idea of her abilities, perhaps based on experience in similar courses or independent reading. 

{Before we proceed with the model, we motivate the assumption that the principal could have more information about the agents than the agents themselves. In such educational examples, the principal could gauge the student's aptitude for the material via a test, and choose how to release the test scores (if at all). Such an assumption applies to more general settings than just education, however. In fact, when one considers the massive amounts of data that platforms collect on users, it is not unreasonable to believe that platform designers are more informed than the users themselves. With recommender systems in particular, websites often collect user data, feed it to a blackbox machine learning algorithm, which then clusters users together into some pre-determined classes. Although the user knows her own online activity (e.g., her Netflix viewing history, or previous Amazon purchases), she does not have access to the class in which the algorithm placed her, or know why a certain product was recommended to her. This is precisely the kind of information asymmetry which our model captures.}

\subsubsection{Information disclosure.} We assume the principal has the ability to commit to an information disclosure policy -- also termed a \emph{signaling scheme} -- which is a mapping from the realized state of nature $\theta$ to a signal of some sort. 
For example, in crowdsourcing platforms, the principal can administer a test to each participant, and choose whether or not (and how) to reveal their scores.
We note that this signaling policy may be randomized. In this paper, we restrict the principal to \emph{public} signaling schemes, i.e., the principal sends the same signal to all agents. Although private signaling schemes have been considered in other settings (for example, signaling in games without externalities~\cite{dughmi2017algorithmic}), it is unclear how an agent can reason about the beliefs and abilities of other agents when evaluating a match.

Formally, the sequence of events is as follows. First, the principal commits to a randomized signaling policy $\phi: \bTheta\to\bSigma$, where $\bSigma$ is the set of all possible signals $\sigma$. Next, the state of nature $\theta\in\bTheta$ is drawn from the prior distribution $\lambda$, and the types of all agents are revealed to the principal. The principal then draws a signal $\sigma$ from the distribution $\phi(\theta)$, and broadcasts it to the agents. The agents in turn compute a posterior on the state of nature $\theta$ given the signal $\sigma$. They compare their optimal match given $\sigma$ and their optimal match without $\sigma$, and choose the agent to match with which maximizes their expected utility. As a result, a matching $m$ is induced. The optimization problem that we are interested in is finding the signaling scheme that maximizes the expected social welfare, i.e., the expected sum of agents' utilities from the resulting matchings.

\subsection{Solution Concept}\label{ssec:sol-concept}

We assume that agents are expected utility maximizers. {Upon receiving the signal from the principal, agents form a posterior on the state of nature $\theta$, given the filtration $\mathcal{F}^i_{\sigma}$ induced by the public signal $\sigma$ and any additional information the agent may have (in particular, the agent's own type in the setting of self-aware agents). A subgame is induced, in which each agent maintains a rank ordered preference list over all other agents, and chooses to match with the agent who maximizes her expected utility, given her available information. Formally, each agent $i$ computes $\EE\left[u(\theta_i,\theta_j)\mid\mF^i_{\sigma}\right]$ for each agent $j\neq i$, and maintains the preference list $j_1 \succcurlyeq j_2 \succcurlyeq \ldots \succcurlyeq \ldots j_{n-1}$ such that  $\EE\left[u(\theta_i,\theta_{j_1})\mid\mF^i_{\sigma}\right]\geq \EE\left[u(\theta_i,\theta_{j_2})\mid\mF^i_{\sigma}\right] \geq \ldots \geq \EE\left[u(\theta_i,\theta_{j_{n-1}})\mid\mF^i_{\sigma}\right]$.}

In the case where agents are self-agnostic, the public signal $\sigma$ induces a \emph{common} filtration $\mathcal{F}_{\sigma}$ among all agents. Consequently, there will be a common ranking of the agents, and the preference profile of each agent is the sub-ranking that does not include her. This preference profile induces a matching $m$, {which leads to the first natural desideratum for endogenous team formation: \emph{stability}}~\cite{gale1962college}.
\begin{definition}
Conditional on available information (priors and signals), a matching $m$ is \emph{stable} if it has no \emph{blocking pair}, i.e., there exist no two agents who are not matched together but would prefer to be matched together.
\end{definition} 
Note that an induced matching $m$ may not be unique, since the ordering is not necessarily strict. We assume that agents break ties in favor of the principal, {as is standard in the Bayesian persuasion literature.}

Though stability is a natural desideratum to impose on team formation, a critical additional requirement is to ensure that agents are willing to participate in any signaling scheme proposed by the principal. Note that in the absence of a signal, agents can still compute $\EE[u(\theta_i,\theta_j)]$, and form a preference profile and resulting matching $m(\emptyset)$ (where $\emptyset$ refers to the absence of signal). This suggests the following natural \emph{endogenous participation} constraint:
\begin{definition}
A signaling scheme is {\emph{Pareto improving}} if it ensures that each agent is weakly better off under any stable matching $m$ induced by the signaling scheme, as compared to her utility under the baseline stable matching $m(\emptyset)$.
\end{definition}

{To summarize, } our aim is to design {Pareto improving} signaling schemes so as to maximize the social welfare under an induced stable matching. {If no such signaling scheme exists, we assume that the principal sends no signal for any realization.}

The following proposition helps further simplify this, by showing that every agent's ranking over other agents induced by a signal $\sigma$ is identical to the ranking induced by the conditional expectation over agent types. We defer its proof to the appendix. 

\begin{prop}\label{prop:util-agent}
	For any $u$ strictly increasing, the expected utility of agent $i$ from matching with agent $j$ is nondecreasing in the expected type of agent $j$, {conditioned on $\mF^i_{\sigma}$}.
\end{prop}

In the absence of signal, the self-agnostic's myopic strategy is obvious. Since $\theta_j\sim Ber(p_j)$ for all $j \in [n]$, we have that $\expectation{\theta_j} = p_j$. Thus, each agent's preference profile results by ordering agents according to $p_j$. Suppose the rank of agents is $a\succcurlyeq b\succcurlyeq c\succcurlyeq d \succcurlyeq\ldots m\succcurlyeq n $, then the resulting teams (formed via stable matching) are $(a,b),(c,d),\ldots,(m,n)$. {Moreover, in a setting with self-agnostic agents, we can extend this argument to see} that under any signal $\sigma$, agents pair up sequentially according to the rank-order of their expected posterior types. {This however is not true for self-aware agents. We return to this latter case in Section~\ref{ssec:self-aware}.}

\subsection{Finding the Optimal Policy}

{The goal of the principal is to find a signaling policy that maximizes the expected \emph{social welfare} (i.e., the sum of all agents' utilities). As discussed in Section~\ref{ssec:sol-concept}, the signaling policy must induce a \emph{stable matching} and satisfy a \emph{Pareto improvement} constraint, such that all agents are better off following the signal than ignoring it. Let $m(\sigma)$ denote the matching induced when the principal sends signal $\sigma$, and $m(\sigma,i)$ denote the agent with whom agent $i$ is matched in matching $m(\sigma)$. Similarly, $m(\emptyset,i)$ is the agent that agent $i$ is matched with in the absence of a signal.}


The principal's optimization problem is given by
\begin{equation}\label{opt-implicit}
\begin{aligned}
& \max_{\phi} & & \expectation{ \sum_{i = 1}^n u\left(\theta_i, \theta_{m(\sigma,i)}\right)}  & \\
& \text{s.t.} & & m(\sigma) \text{ is stable} &  \forall \; \sigma\in\bSigma\\
& & & \EEc{u(\theta_i,\theta_{m(\sigma,i)})}{\mathcal{F}^i_\sigma} \geq \EEc{u(\theta_i,\theta_{m(\emptyset,i)})}{\mF^i_\emptyset} & \forall \; \sigma\in\bSigma,\forall \; i \in [n]
\end{aligned}
\end{equation}
where the maximization is over the set of all randomized maps $\phi$ from type-profiles $\theta$ to signals $\sigma$, and the expectations are taken over the joint distribution of $(\sigma, \theta)$ induced by $\phi$. 

A standard revelation-principle style argument tells us that there exists an optimal signaling policy that is \emph{straightforward} and \emph{persuasive} \cite{kamenica2011bayesian,dughmi2017algorithmic,arieli2016private}. A \emph{persuasive} signaling policy is a policy for which the induced actions form a Bayesian Nash equilibrium. A \emph{straightforward} signaling policy corresponds to space of signals equaling the agents' action space. In our setting, a straightforward signaling scheme is one that signals a matching for a given realization $\theta$. Additionally, by Proposition~\ref{prop:util-agent}, matchings are induced by a rank ordered preference list of agents' expected posterior types. Combining these three facts, we get that \emph{it is sufficient to restrict our attention to signals that are rank-orderings of agents according to expected posterior type}.

We first need some additional notation. Let $\bSigma$ now denote the set of all orderings (or permutations) over the $n$ agents. We use the notation $\sigma_{i}$ to denote the identity of the $i$th-placed individual with respect to the ordering $\sigma$. Additionally, $\sigma_i(\theta)$ denotes the type of agent $\sigma_{i}$ for the specific realization $\theta$. We abuse notation and use $\sigma_{m(i)}$ to denote the agent that the $i$th agent in ordering $\sigma$ is matched to. Finally, $\sigma_{m(i)}(\theta)$ denotes the type of the matchof agent $i$ under ordering $\sigma$ of a given realization $\theta$. The following example helps clarify this notation. ({Henceforth, we will use ``match'' to denote both the team an agent is a part of, as well as the agent's teammate. The specific use will be clear from context.}) 

\begin{example}
Suppose $n=4$, with agents A, B, C and D. Consider realization $\theta = (\theta_A,\theta_B,\theta_C,\theta_D) = (1, 1, 0, 0)$, and $\sigma = (A\succcurlyeq C \succcurlyeq B \succcurlyeq D)$. Recall that for a given rank-order, the agents pair up sequentially. Thus, $\sigma(\theta) = (1,0,1,0)$, $\sig{2} = $ C, $\sig{2}(\theta) = 0$. Similarly, $\sig{m(2)} = $ A, and $\sig{m(2)}(\theta) = 1$.
\end{example}


Proposition~\ref{prop:matching} gives us a tractable way of representing a persuasive signal. We defer its proof to the appendix.

\begin{prop}\label{prop:matching}
	The matching induced by the signal (i.e., the announced ordering) is stable if and only if the signal is persuasive, i.e., if and only if the following holds.
	\begin{align}\label{eq:correct-ordering}
	 \EEc{\sig{m(i)}(\theta)}{\mathcal{F}^i_{\sigma}} \geq \EEc{\sigma_j(\theta)}{\mathcal{F}^i_{\sigma}} \qquad \forall \; i \in [n-1], \forall j \text{ s.t. } \sigma_j > \sig{m(i)}, \sigma_j\neq\sigma_i.
	\end{align} 
\end{prop}



Putting this all together, we obtain a concise representation of the principal's optimization problem as follows:
\begin{equation}\label{opt}\tag{$OPT$}
\begin{aligned}
& \max_{\phi} & & \EE\left[\sum_{i=1}^n u\left(\sigma_i(\theta), \sig{m(i)}(\theta)\right)\right]   \\
& \text{s.t. } & & \EEc{\sig{m(i)}(\theta)}{\mathcal{F}^i_{\sigma}} \geq \EEc{\sig{j}(\theta)}{\mathcal{F}^i_{\sigma}}  \hspace{0.5cm} \forall \; \sigma\in\bSigma, \forall \; i \in [n-1], \forall \; j \text{ s.t. } \sigma_j > \sig{m(i)}, \sigma_j \neq \sigma_i\\
& & &\EEc{u(\theta_i,\theta_{m(\sigma,i)})}{\mathcal{F}^i_\sigma} \geq \EEc{u(\theta_i,\theta_{m(\emptyset,i)})}{\mF^i_\emptyset}\hspace{0.5cm}   \forall \; \sigma\in\bSigma,\forall \; i \in [n].
\end{aligned}
\end{equation} 

\subsection{Classes of utility functions}

{We will see that the structure of the optimal signaling scheme depends on (an appropriate notion of) the convexity of agents' utility function $u$. In particular, we use the following notion of convexity (resp., concavity) in our subsequent results.}

	\begin{definition}
		\begin{enumerate}[(i)]
			\item We say that $u$ is \emph{convex} if $u(1,1) - u(1,0) \geq u(1,0) - u(0,0)$.
			\item We say that $u$ is \emph{strictly concave} if $u(1,1) - u(1,0) < u(1,0) - u(0,0)$.
		\end{enumerate}
	\end{definition}	
	
	We note that, although this definition differs from multivariate notions of discrete convexity, it is equivalent to discrete convexity/concavity of $u$ if viewed as a univariate function of the sum of agent-types. {More generally, this admits any $u$ of the form $g\left(\sum_if(\theta_i)\right)$, where $g$ is strictly convex/concave, and $f$ increasing. Interpreting  $u$ as the likelihood of successfully completing the task, one can view the above notion of convexity as the relative benefits of teaming up to complete the task -- if $u$ is convex, then the task is difficult enough that the team greatly benefits from having more highly skilled individuals attempt it; on the other hand, if $u$ is concave, then the task is simple enough that additional skilled individuals contribute only marginally. \\
	
	Lemma~\ref{lem:convexity} tells us that, to maximize social welfare, it suffices to consider the \emph{number} of matches of a certain type, and that the preferred match type depends on the convexity of $u$. We defer its proof to the appendix.
	
	In the remainder of the paper, we will use $m_{11}, m_{10}, m_{00}$ to denote the number of `1-1', `1-0', `0-0' matches, respectively.
	
	\begin{lemma}\label{lem:convexity}
		\begin{enumerate}[(i)]
			\item Suppose $u$ is convex. Then, maximizing the expected social welfare is equivalent to maximizing the expected number of `1-1' matches (equiv., maximizing the expected number of `0-0' matches or minimizing the number of `1-0' matches).
			\item Suppose $u$ is strictly concave. Then, maximizing the expected social welfare is equivalent to maximizing the expected number of `1-0' matches (equiv., minimizing the number of `0-0' matches or the number of `1-1' matches).
		\end{enumerate}
	\end{lemma}

\section{Self-Agnostic Agents}
\label{ssec:agnostic}

We first consider the setting where agents are \emph{self-agnostic}, i.e., they do not know their own types. We provide optimal signaling schemes when $(i)$ agent utilities are convex,  and $(ii)$ agents have i.i.d. types and their utilities are strictly concave. We the generalize the latter to the setting with multiple clusters. There, we provide a signaling scheme which is asymptotically optimal in $n$ (for a fixed number of clusters).

When agents do not have access to their own types, they share a common posterior induced by the announced ordering $\sigma$. Thus, the set of `persuasiveness' constraints in Problem~\eqref{opt} reduces to:
$$ \EEc{\sig{i}(\theta)}{\sigma} \geq \EEc{\sig{i+1}(\theta)}{\sigma} \qquad \forall\;\sigma\in\bSigma, \forall \; i \in [n-1].$$

Note that the above constraints require that, for two agents matched together, in expectation the agent who comes first in the ranking must indeed have a higher type than the agent who comes second. Adding this constraint is without loss of generality due to the assumption that agents' utilities are {symmetric} and anonymous. 

\subsection{Convex utilities}

Our first main result is that, for convex agent utilities, the principal's optimization problem is, in some sense, easy. 
	
	\begin{theorem}\label{thm:convex-truth}
		If $u$ is convex, then the optimal signaling policy is to always announce the true ordering.
	\end{theorem}
		
	\begin{proof}
			By Lemma~\ref{lem:convexity}, for convex $u$, the principal's optimization problem reduces to maximizing the expected number of `1-1' matches. Suppose announcing the true ordering of agents is Pareto-improving. Then, this scheme is clearly optimal since, for each realization $\theta$, all type 1 individuals will pair up (modulo one odd-one-out), hence maximizing $m_{11}$ for each realization. The resulting matching for each realization will be stable, as the truly highest-ranked agents will be satisfied with their match. Thus, it remains for us to show that this scheme is Pareto-improving. 
		
		Since agent types are binary, the true ordering is not unique, and the policy is thus underspecified. Our policy is as follows. Start with the baseline (myopic) ordering, in which agents only match within their own cluster. Conserve all `1-1' and `0-0' matches. Pick one `1-0' match uniformly at random to conserve. Split up the rest to create all `1-1' and `0-0' matches. (This is equivalent to taking the randomly chosen `1-0' match, putting that type 1 agent to be the last type 1 agent announced in the ordering, and that type 0 agent to be the first type 0 agent announced in the ordering.) 		
		 		
		Let $\mathcal{S}^\text{truth}$ denote the truthful, fair signaling scheme described above, and $\mathcal{U}_i(\mathcal{S}^\text{truth})$ denote agent $i$'s expected utility under this scheme. Let $j$ denote agent $i$'s baseline match, and suppose $\theta_i, \theta_j \sim Ber(p)$. Moreover, let $E$ denote the event that the $(i,j)$ match is chosen to be the mixed match.  We have: 		
		\begin{align*}
		\mathcal{U}_i(\mathcal{S}^\text{truth}) &= p^2 u(1,1) + (1-p)^2 u(0,0) \\
		&\hspace{1cm}+ p(1-{p}) \left(u(1,1)\left(1-\mathbb{P}\left[E \mid \theta_i = 1, \theta_j = 0\right]\right) + u(1,0) \mathbb{P}\left[E \mid \theta_i = 1, \theta_j = 0\right] \right)\\
		&\hspace{1cm}+(1-p){p}\left(u(1,0)\mathbb{P}\left[E \mid \theta_i = 0, \theta_j = 1\right] + u(0,0)\left(1-\mathbb{P}\left[E \mid \theta_i = 0, \theta_j = 1\right]\right)\right).
		\end{align*}

		
		Since a `1-0' match is chosen to be conserved uniformly at random, it is easy to see that $\mathbb{P}\left[E \mid \theta_i = 1, \theta_j = 0\right] = \mathbb{P}\left[E \mid \theta_i = 0, \theta_j = 1\right]$. We denote the probability of this event $\alpha$. Thus,
	\begin{align*}
		\mathcal{U}_i(\mathcal{S}^\text{truth})  - \mathcal{U}_i(\emptyset) &=  p^2 u(1,1) + (1-p)^2 u(0,0) \\
		&\hspace{1cm}+ p(1-{p}) \left(u(1,1)\left(1-\alpha\right) + u(1,0) \alpha \right)
		\\
		&\hspace{1cm}+(1-p){p}\left(u(1,0)\alpha + u(0,0)\left(1-\alpha\right)\right) \\
		&\hspace{1cm}-\left(p^2 u(1,1) + (1-p)^2 u(0,0) + u(1,0) \left(p(1-{p}) +{p}(1-p) \right)\right)  
		\\
		&= p(1-{p})\left(1-\alpha\right)\left(\left(u(1,1)-u(1,0)\right) - \left(u(1,0)-u(0,0)\right)\right)  \\
		&\geq 0 \text{ by convexity}
		\end{align*}

	\end{proof}

In the setting with convex utilities, the principal's and type 1 agents' incentives are aligned, since it is welfare-optimal for type 1 agents to be matched together. In this case, then, stability is easy to achieve; the tension that dominates is that of endogenous participation for type 0 agents. By convexity, however, the additional gain from being in a `1-1' match outweighs the loss from being assigned to a `0-0' match.

We note that the uncertainty in types is key here; we will later see that this is not the case when agents are self-aware. 
	
\subsection{Concave utilities}\label{ssec:self-ag-concave}
	
We now consider the more interesting case, where utilities are \emph{strictly concave}. By Lemma~\ref{lem:convexity}, maximizing social welfare is equivalent to minimizing the number of `0-0' matches. In this setting, the principal's and type 1 agents' incentives are now misaligned. Thus, not only is there the difficulty of ensuring endogenous participation of type 1 agents, but stability also becomes an inherently harder desideratum to achieve. We approach these challenges by first tackling the problem of implementing a stable matching.

Let $x_{\theta,\sigma}$ denote the probability of announcing ordering $\sigma$ when the realization is $\theta$, and let $m_{00}(\sigma(\theta))$ be the resulting number of `0-0' matches under this realization and announced ordering. {Similarly, $m_{11}(\sigma(\theta)), m_{10}(\sigma(\theta))$ respectively denote the number of `1-1', `1-0' matches induced by announcing $\sigma$ for realization $\theta$.} {Further, recall that $\lambda(\theta)$ is the prior likelihood of observing $\theta$.}  We consider the relaxed problem which ignores the Pareto improvement constraints of agents, and is thus a lower bound on (\ref{opt}): 
\begin{equation}\tag{$\widehat{OPT}$}\label{relaxed-lp}
\begin{aligned}
&\min_{\{x_{\theta,\sigma}\}} & & \sum_{\theta\in\bTheta}\sum_{\sigma\in\bSigma} \lambda(\theta)x_{\theta,\sigma}m_{00}(\sigma(\theta))  & \\
& \text{s.t. } & & \sum_{\theta\in\bTheta}\lambda(\theta)x_{\theta,\sigma}\sig{i}(\theta) \geq \sum_{\theta\in\bTheta}\lambda(\theta)x_{\theta,\sigma}\sig{i+1}(\theta) &   \forall \; \sigma\in\bSigma, \forall \; i \in [n-1]\\
&   & & \sum_{\sigma\in\bSigma}x_{\theta,\sigma} = 1 & \forall \; \theta\in\bTheta\\
& & &x_{\theta,\sigma} \geq 0 & \forall \; \theta \in \bTheta, \forall \; \sigma\in\bSigma
\end{aligned}
\end{equation}

Let \eqref{relaxed-lp} denote the optimal value of the relaxed linear program. This LP is computationally intractable, even for small values of $n$, since it is exponentially large in both the number of variables (of which there are $n!2^n$) and the number of constraints (of which there are $O(n!2^n)$). {Before proceeding, we note that by Lemma~\ref{lem:convexity} we could have equivalently focused on maximizing the expected number of `1-0' matches. For our high-level intuition throughout the rest of the paper, we will refer to these two objectives interchangeably.}

\subsubsection{Warm-up: Single cluster}
Despite the complexity of the above LP, for the special case of i.i.d. agent types (i.e., where all agents form a single cluster with the same prior on their types), we are able to demonstrate an \emph{optimal signaling scheme}. Moreover, the resulting social welfare matches the \emph{first-best} solution, i.e., the utility of the optimal matching ignoring the Pareto improvement and persuasiveness constraints. 

\begin{theorem}\label{thm:fb-symm}
	In the single-cluster setting where all agents have i.i.d. types and strictly concave utilities, for any realized type-profile $\theta$, the principal can achieve $m_{00}(\sigma^*(\theta)) = \min\limits_{\sigma}m_{00}(\sigma(\theta))$. Thus,
	\begin{align*}
	OPT = \expectation{\min_{\sigma}m_{00}(\sigma(\theta))}.
	\end{align*}
\end{theorem}

We provide a proof sketch of Theorem~\ref{thm:fb-symm} below, and defer the formal proof to Appendix~\ref{sec:app-thm2-proof}.

\begin{proof}[Proof outline]
To prove the theorem, we first solve the relaxed problem (\ref{relaxed-lp}), and show that, for all realizations $\theta$, the principal can construct a scheme such that the number of matches between two type 0 agents is as small as possible. We henceforth refer to this scheme as \emph{First Best} (denoted $FB$) since this is the best achievable utility without strategic considerations. We then show that this signaling scheme satisfies the Pareto improvement constraint for all agents, and hence is feasible for (\ref{opt}). We informally describe the scheme.
	
The main intuition behind designing our persuasive signaling scheme to achieve the first-best outcome is to `pair' type-profiles for a given signal $\sigma$ such that together they satisfy the rank-order in the signal. In particular, since agents are self-agnostic, we can leverage this by pairing up type-profiles such that for each type-profile $\theta$ under which two agents of different types are matched, there exists a profile $\bar\theta$, realized with equal probability, in which the same agents are matched, but have their types flipped. This ensures that the incentive constraints in (\ref{relaxed-lp}) are satisfied, as we are incentivizing a strong agent  to accept being matched with a weak agent via a promise of matching with a strong agent when they themselves are weak. 
	
	We illustrate the \emph{First Best} construction via an example.
	
	\begin{example}\label{ex:fb}
		Consider $\theta = (\theta_A, \theta_B, \theta_C, \theta_D, \theta_E, \theta_F) = (1, 1, 1, 1, 0, 0)$. Then, $\sigma^* = (A\succcurlyeq B \succcurlyeq C  \succcurlyeq E  \succcurlyeq D  \succcurlyeq F)$ achieves zero `0-0' matches. Consider $\bar\theta = (1, 1, 0, 0, 1, 1)$. Announcing the same signal $\sigma^*$ for this realization also achieves zero `0-0' matches. Additionally, note that since $\J(\theta) = \J(\bar\theta)$ and agents have i.i.d. types, $\lambda(\theta) = \lambda(\bar\theta)$. 
		Suppose that, given either of these realizations (and for no other realization), the principal announces $\sigma^*$ with probability 1. Then, when agents receive the signal $\sigma^*$, the posterior probability of being in realization $\theta$ is equal to the posterior probability of being in $\bar\theta$. Thus, the expected posterior types of agents $A$ and $B$ is 1; the expected posterior type of all other agents is 0.5, and the persuasiveness constraints are thus respected.
	\end{example}

	Finally, the fact that the above scheme satisfies the participation constraints depends on the concavity of the agents' utility function. The intuition behind this is that under a strictly concave $u$, the additional value of a `1-1' match compared to a `1-0' match is less than the value of being in a `1-0' match than being in a `0-0' match. Consequently, having an assurance of being in a 0-0 team as rarely as possible dominates the match under the myopic matching.

\end{proof}

To summarize, Theorem~\ref{thm:fb-symm} states that when agents have i.i.d. types and are agnostic of their own types, the principal has enough freedom to match as many type 1 agents with type 0 agents as possible, and to do so such that (i) the declared rank-ordering is persuasive and (ii) agents are better off under this mechanism than under their myopic strategy. In the next section, we use the \emph{First Best} signaling scheme as a critical primitive for the multi-cluster setting.

\subsubsection{Multiple clusters}


Our first observation is that, in the setting with multiple clusters, there exist instances for which there is no hope of achieving first-best (i.e., maximal pairing of type 1 and type 0 agents).  We will see that this fact creates a challenge in analyzing our signaling scheme, as first-best is far too loose a bound. 

\begin{prop}
For multiple clusters of agents, the gap between $OPT$ and first-best is unbounded.
\end{prop}

\begin{proof}
Consider the following trivial example of $K=2$ clusters, with $p_1 = 1, p_2 = 0$ and $n_1 = n_2$. Given that Cluster 1 agents know with certainty that they are type 1 agents, no signal can incentivize them to match with Cluster 2 agents. The optimal solution, then, is to send no signal, thereby achieving $\frac{n_2}/2$ `0-0' matches.The first-best scheme, on the other hand, obtains zero `0-0' matches.
\end{proof}

Consider the scheme that only matches agents within their own clusters, and does so using the \emph{First Best} signaling scheme from the single-cluster setting; we refer to this scheme as \emph{Cluster First Best}, denoted as $FB_C$. Given Theorem~\ref{thm:fb-symm}, this scheme is feasible, and thus gives us an upper bound on $OPT$. However, the scheme is sub-optimal for $K\geq 2$ clusters, since it misses opportunities to match excess type 1 and type 0 agents across clusters, an event which occurs a non-trivial number of times. 

Despite this, as the number of agents $n$ scales, while keeping the number of clusters constant, we show that the \emph{Cluster First Best} policy gets arbitrarily close to the optimal solution for $K=2$ clusters.
\begin{theorem}\label{thm:2-clusters}
	For $K = 2$ clusters and $n_1, n_2$ such that $\sqrt{\frac{\ln n_1}{n_1}} \leq |p_1-1/2|,\sqrt{\frac{\ln n_2}{n_2}} \leq |p_2-1/2|$, then Cluster First Best achieves $o(1)$ regret (with respect to $n$).
\end{theorem}

As noted above, \emph{Cluster First Best} is suboptimal since it misses occasions where it could induce agents of different types to match \emph{across} different clusters. The first important realization is that the \emph{only} realizations for which Cluster First Best is suboptimal are realizations for which there is an excess of type 1 agents in one cluster, and an excess of type 0 agents in the other cluster. For all other realizations, Cluster First Best is in fact optimal. 


For the realizations where there are potential gains from matching agents across clusters, we can no longer use the first-best solution as a benchmark. To show that the $FB_C$ policy is asymptotically optimal in these cases is much more challenging. We provide a brief sketch of this proof below, and defer the formal proof to Appendix~\ref{ssec:thm-3-proof}. 
\begin{proof}[Proof outline]
	The proof of the theorem considers three regions of the $(p_1, p_2)$ space:\\
	$(i)$ $p_1 > p_2 > 1/2$\quad,\quad
	$(ii)$ $p_2 < p_1 < 1/2$\quad,\quad
	$(iii)$ $p_2 < 1/2 < p_1$.
	
	{We note that assuming $p_1 > p_2$ is without loss of generality since we can relabel the clusters.}
	
	Via measure concentration arguments, we can restrict ourselves to considering realizations that lie in a `typical set' (following standard information theoretic definitions; cf.~\cite{cover2012elements}): $$\typical{2} \triangleq \left\{\theta: |\J_1(\theta)-n_1p_1|\leq \epsilon_1n_1,|\J_2(\theta)-n_2p_2|\leq \epsilon_2n_2 \right\}.$$
	Note that for sufficiently large $n$, {with high probability} the realized type-profiles in $\typical{2}$ are such that $(h_1 (\theta), h_2(\theta))$, the number of type 1 agents in each cluster, is in the same orthant as $(n_1p_1, n_2p_2)$, respectively.
	
	\begin{figure}[h]
		\centering
		\includegraphics[scale=0.6]{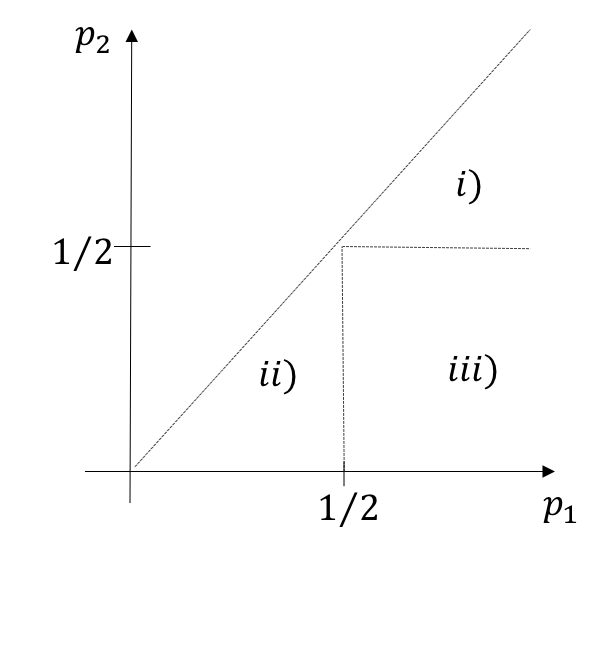}
		\caption{}
		\label{fig:proof-sketch}
	\end{figure}

	Our proof for cases $(i)$ and $(ii)$ relies on the observation that in these regions, the induced matching under \emph{Cluster First Best is exactly First Best}. In Case $(i)$, for example, all realizations in $\typical{2}$ will have an excess of type 1 agents in each cluster. Thus, if we simply match optimally \emph{within} each cluster, only type 1 agents remain to be matched, and we could not have done better. A similar argument holds for Case $(ii)$. We finish our argument by using concentration inequalities for large enough $n_1, n_2$.
	
	The main technical challenge in the proof is in dealing with Case $(iii)$. Here, we exhibit a dual certificate solution that gives a lower bound on the optimal solution for the restricted space $\typical{2}$. In particular, we construct a feasible solution for the dual LP of (\ref{relaxed-lp}) whose value is exactly that of \emph{Cluster First Best}.
	The construction relies on an intricate inductive argument, the details of which we defer to Appendix~\ref{ssec:thm-3-proof}. 
	Finally, via concentration arguments, we show that, in going from the typical set to the entire space of realizations, the lower bound decreases at most by $o(1)$ for large enough $n_1, n_2$.\end{proof}

Theorem~\ref{thm:2-clusters} can be extended to provide a similar result for $K > 2$ clusters.

\begin{corollary}\label{cor:k-clusters}
	For $K > 2$ clusters with $n_1,\ldots,n_K$ agents, such that $\sqrt{\frac{\ln n_1}{n_1}} \leq |p_1-1/2|,\ldots,\sqrt{\frac{\ln n_K}{n_K}} \leq |p_K-1/2|$, Cluster First Best achieves $o(1)$ regret.
\end{corollary}
The proof of the corollary requires an additional induction on the number of clusters. The high-level idea remains: for large clusters, the number of type 1 agents in each cluster is in the $(n_1p_1,\ldots,n_Kp_K)$ orthant with high probability. Thus, there do not exist enough realizations in other orthants that would `fix' the excesses in this high probability orthant. We defer the formal proof to Appendix~\ref{ssec:thm-3-proof}.

Theorem~\ref{thm:2-clusters} and Corollary~\ref{cor:k-clusters} tell us that, though sophisticated information disclosure policies can induce agents to pair with agents in other clusters, for a large enough number of agents, the fraction of type 1  agents in the `stronger' clusters concentrates so much so that the principal has no leverage to induce them to match with an agent in a weaker cluster. Thus, the principal does not gain a significant amount from more sophisticated signaling schemes, and \emph{Cluster First Best} is close to optimal.

\section{Self-Aware Agents}\label{ssec:self-aware}


We now turn our attention to settings where agents are aware of their own types.  We use the same solution concepts as for self-agnostic agents: Pareto improvement and stability. In this case, however, the filtration induced by the signal is no longer common to all agents, since they have the additional knowledge of their own type. Consequently, stability is a much more difficult goal to achieve in this scenario.

Following a similar argument as in Section~\ref{ssec:agnostic}, the principal's optimization problem in the setting with self-aware agents can be written as:
\begin{equation}
\begin{aligned}
& \max_{\{x_{\theta,\sigma}\}} & & \expectation{\sum_{i=1}^n u(\sig{i}(\theta),\sig{m(i)}(\theta))}  & \\
& \text{s.t. } & & \sum_{\theta:\sig{i}(\theta) = 1}\lambda(\theta)x_{\theta,\sigma}\sig{m(i)}(\theta) \geq \sum_{\theta:\sig{i}(\theta) = 1}\lambda(\theta)x_{\theta,\sigma}\sigma_{j}(\theta) &  \forall \; \sigma\in\bSigma, \forall \; i\in[n-1], \forall \; j > m(i) \\
& & & \sum_{\theta:\sig{i}(\theta) = 0}\lambda(\theta)x_{\theta,\sigma}\sig{m(i)}(\theta) \geq \sum_{\theta:\sig{i}(\theta) = 0}\lambda(\theta)x_{\theta,\sigma}\sigma_{j}(\theta) &  \forall \; \sigma\in\bSigma, \forall \; i\in[n-1], \forall \; j > m(i) \\
& & &\EEc{u(\theta_i,\theta_{m(\sigma,i)})}{\theta_i,\sigma} \geq \EEc{u(\theta_i,\theta_{m(\emptyset,i)})}{\theta_i} & \forall \; \sigma\in\bSigma,\forall \; i \in [n], \theta_i\in\{0,1\}\\
&   & & \sum_{\sigma\in\bSigma}x_{\theta,\sigma} = 1 & \forall \; \theta\in\bTheta\\
& & &x_{\theta,\sigma} \geq 0 & \forall \; \theta \in \bTheta, \forall \; \sigma\in\bSigma\label{eq:known-types}
\end{aligned}
\end{equation}

Intuitively, the principal has less leverage when agents know their own types than in the scenario with self-agnostic agent. {In the case with strictly concave utilities, for example, since a type 1 agent knows that she has a high type, at a high level she feels entitled to be matched with another type 1 agent. The difficulty for the principal arises because any recommendation made should further the goal of creating more `1-0' matches, which is an indicator to the type 1 agent that she is getting a `bad deal', in some sense.}

Additionally, from a computational point of view, the optimization problem appears more difficult to solve since the information disclosure policy now has to take into account the different priors that agents have over one another. 
This point however turns out to be moot, as we show below that this setting suffers from a strong impossibility result, for both classes of utility functions.

\begin{theorem}
\label{thm:known-types-impossibility}
When agents are self-aware, no {Pareto improving} signaling policy can perform better than random \emph{for both convex and strictly concave utilities}. Moreover
\begin{enumerate}[(i)]
\item For convex utilities, the agents blocking any other signaling scheme are type 0 agents.
\item For strictly concave utilities, the agents blocking any other signaling scheme are type 1 agents.
\end{enumerate}
\end{theorem}

\begin{proof}
We first show the claim in the case of strictly concave utilities. 
	
Consider cluster $k$, in which each agent has i.i.d. type drawn from a $Ber(p_k)$ distribution. Recall that $n_k$ denotes the number of agents in cluster $k$.

Now consider any signaling policy $\mathcal{S}$. The Pareto improvement constraint imposes that the expected utility of each agent $i$ under $\mathcal{S}$ weakly improve upon her expected utility when matched with a random agent in her cluster.
	
	Suppose agent $i$ is a type 1 agent, and let $k(i)$ denote her cluster. Under her myopic strategy, agent $i$'s expected utility is:
	\begin{align*}
	\mathcal{U}_i(\emptyset)&= p_{k(i)}u(1,1)+(1-p_{k(i)})u(1,0)= u(1,0)+p_{k(i)}\left(u(1,1)-u(1,0)\right).
	\end{align*}
	
	Let $p_{\mathcal{S}}^{i}$ denote the probability that agent $i$ is matched with another type 1 agent under signaling policy $\mathcal{S}$, given that she is a type 1 agent. Similarly, under $\mathcal{S}$, agent $i$'s expected utility is:
	\begin{align*}
	\mathcal{U}_i(\mathcal{S}) = u(1,0)+p_{\mathcal{S}}^{i}\left(u(1,1)-u(1,0)\right).
	\end{align*}
	
	To satisfy the Pareto improvement constraint, $\mathcal{S}$ must satisfy:
	\begin{align}\label{eq:aware1}
	&u(1,0)+p_{\mathcal{S}}^{i}\left(u(1,1)-u(1,0)\right) \geq u(1,0)+p_{k(i)}\left(u(1,1)-u(1,0)\right) \notag \\
	\iff &p_{\mathcal{S}}^{i} \geq p_{k(i)}
	\end{align}
	
	By Lemma~\ref{lem:convexity}, for strictly concave $u$, maximizing social welfare is equivalent to minimizing the expected number of `1-1' matches.
	We expand $m_{11}\left(\sigma(\theta)\right)$ as:
	\begin{align*}
	m_{11}(\sigma(\theta)) = \frac12\sum_{i=1}^n\mathds{1}\{\sig{i}(\theta)=1,\sig{m(i)}(\theta)=1\}.	\end{align*} 
	Plugging this into our objective, we obtain:
	\begin{align}
	\sum_{\theta\in\bTheta}\sum_{\sigma\in\bSigma}\lambda(\theta)x_{\theta,\sigma}\sum_{i=1}^n\frac12\mathds{1}\{\sig{i}(\theta)=1,\sig{m(i)}(\theta)=1\}&= \frac12\sum_{i=1}^n p_{k(i)} \left(\sum_{\theta:\theta_i = 1}\sum_{\sigma}\lambda(\theta)x_{\theta,\sigma}\mathds{1}\left\{\sig{m(i)}(\theta)=1\right\}\right)\nonumber\\
	&= \frac12  \left(\sum_{i=1}^n p_{k(i)}  p_{\mathcal{S}}^{i}\right) \geq \frac12  \left( \sum_{i=1}^n p_{k(i)}^2 \right)= \sum_k \frac{n_kp_{k}^2}{2} \nonumber
	\end{align}
	where the inequality follows from~\eqref{eq:aware1}. This lower bound is achieved by having the agents match randomly within each cluster, their myopic strategy. 
%

We now show the claim in the case of convex utilities.  	
	Suppose agent $i$ is a type 0 agent. Under her myopic strategy, agent $i$'s expected utility is:
	\begin{align*}
	\mathcal{U}_i(\emptyset)  &= p_{k(i)}u(1,0)+(1-p_{k(i)})u(0,0).
	\end{align*}
	
	Let $p_{\mathcal{S'}}^{i}$ denote the probability that agent $i$ is matched with another type 0 agent under signaling policy $\mathcal{S'}$, given that she is a type 0 agent. Under $\mathcal{S'}$, agent $i$'s expected utility is:
	\begin{align*}
	\mathcal{U}_i(\mathcal{S}) &= u(1,0)+p_{\mathcal{S'}}^{i}\left(u(0,0)-u(1,0)\right).
	\end{align*}
	
	To satisfy the agent rationality constraint, $\mathcal{S}$ must satisfy:
	\begin{align}\label{eq:aware2}
	&u(1,0)+p_{\mathcal{S'}}^{i}\left(u(0,0)-u(1,0)\right) \geq p_{k(i)}u(1,0)+(1-p_{k(i)})u(0,0) \notag \\
	\iff &p_{\mathcal{S'}}^{i} \leq 1-p_{k(i)}
	\end{align}
	
	By Lemma~\ref{lem:convexity}, for convex $u$, maximizing social welfare is equivalent to maximizing the expected number of `0-0' matches.
	We expand $m_{00}\left(\sigma(\theta)\right)$ as:
	\begin{align*}
	m_{00}(\sigma(\theta)) = \frac12\sum_{i=1}^n\mathds{1}\{\sig{i}(\theta)=0,\sig{m(i)}(\theta)=0\}.	
	\end{align*} 
	Plugging this into our objective of interest, we obtain:
	\begin{align}
	\sum_{\theta\in\bTheta}\sum_{\sigma\in\bSigma}\lambda(\theta)x_{\theta,\sigma}\sum_{i=1}^n\frac12\mathds{1}\{\sig{i}(\theta)=0,\sig{m(i)}(\theta)=0\}&= \frac12\sum_{i=1}^n(1-p_{k(i)}) \left(\sum_{\theta:\theta_i = 0}\sum_{\sigma}\lambda(\theta)x_{\theta,\sigma}\mathds{1}\left\{\sig{m(i)}(\theta) = 0\right\}\right)\nonumber\\
	&= \frac12 \left(\sum_{i=1}^n (1-p_{k(i)})  p_{\mathcal{S'}}^{i}\right) \leq \frac12  \left( \sum_{i=1}^n (1-p_{k(i)})^2 \right) \nonumber \\
	&= \sum_k\frac{n_k(1-p_{k})^2}{2}\nonumber
	\end{align}
	where the inequality follows from~\eqref{eq:aware2}. This upper bound is achieved by having the agents match randomly within each cluster, their myopic strategy. 
\end{proof}

\section{Teams of arbitrary size}\label{ssec:larger-size}

In this section, we show how our results extend to teams of constant size $a > 2$, where we assume that the size of each cluster is a multiple of $a$. The setup is mostly analogous to that of two-member teams. As such, we only highlight fundamental modeling differences.

The first such fundamental difference is the solution concept, which, in the setting with larger teams, is that of \emph{group stability}, as defined in~\cite{bogomolnaia_stability_2002}. We abuse notation and let $\mathcal{U}_i(t)$ denote the expected utility of agent $i$ obtains from being in team $t$, conditional on all available information (priors and signals).
\begin{definition}
A set of teams is \emph{group stable} if there exists no two teams $t, t'$, agents $i \in t, j \in t'$ such that the following holds:
\begin{align*}
\begin{cases}
&\mathcal{U}_j\left(t\cup \{j\} \setminus \{i\}\right) > \mathcal{U}_j(t')\\
&\mathcal{U}_k\left(t\cup \{j\} \setminus \{i\}\right) > \mathcal{U}_k(t) \quad \forall k \in t, k \neq i
\end{cases}
\end{align*}
\end{definition}

We also need an appropriate generalization of convexity. 
	\begin{definition}
	Let $u(j)$ denote the utility gained from having $j$ type 1 agents in a team. 
		\begin{enumerate}[(i)]
			\item $u$ is \emph{convex} if $u(j+1) - u(j) \geq u(j) - u(j-1)$.
			\item $u$ is \emph{strictly concave} if $u(j+1) - u(j) < u(j) - u(j-1)$.
		\end{enumerate}
	\end{definition}
	
	Just as Lemma~\ref{lem:convexity} relates the social welfare maximization problem with the problem of maximizing (resp. minimizing) the number of matches of a certain type, the following lemma gives an equivalent formulation of our problem for arbitrary team size.  
		
	\begin{lemma}
		Consider the following objective:
		\begin{align*}
		OBJ \triangleq \sum_{\theta}\sum_{\sigma}\lambda(\theta)x_{\sigma \theta} \Big[ &\alpha m_{a/2}(\sigma(\theta)) + \beta\left(m_{a/2 - 1}(\sigma(\theta)) + m_{a/2 + 1}(\sigma(\theta)) \right) + \ldots\\
		&\ldots + \gamma\left(m_a(\sigma(\theta)) + m_0(\sigma(\theta))\right)\Big],
		\end{align*}
		where $m_j(\sigma(\theta))$ denotes the number of teams with $j$ type 1 agents when ordering $\sigma$ is announced for realization $\theta$,  and $(\alpha,\beta,\ldots,\gamma)$ are in the simplex.
		\begin{enumerate}[(i)]
			\item Suppose $u$ is convex. Then, maximizing the expected social welfare is equivalent to maximizing $OBJ$, with $\alpha < \beta < \ldots < \gamma$.
			\item Suppose $u$ is strictly concave. Then, maximizing the expected social welfare is equivalent to maximizing $OBJ$, with $\alpha > \beta > \ldots > \gamma$ or the number of `1-1' matches).
		\end{enumerate}
	\end{lemma}
	
	We leave the proof of the lemma to the reader, as it is analogous to that of Lemma~\ref{lem:convexity}. 

\subsection{Self-agnostic agents}

{\subsubsection{Convex utilities.}

\begin{theorem}\label{thm:convex-truth-teams}
	For convex utility functions and team size $a > 2$, the optimal signaling policy is to always announce the true ordering.
\end{theorem}

\begin{proof}
	As in the setting with teams of size 2, announcing the true ordering is trivially persuasive. Thus, it remains for us to show that there exists a Pareto-improving full-information scheme, and this policy would hence be optimal.
	
	Consider the full-information scheme which breaks ties uniformly at random \emph{within clusters}. That is, the scheme treats agents within each cluster identically, but is allowed to differentiate between clusters. Pareto improvement of this scheme is a corollary to the following lemma.
	
	\begin{lemma}\label{lem:welf-opt}
		In the self-agnostic setting, given $K$ clusters, any utility function $u$ and team size $a > 2$, any signaling scheme $\mathcal{S}$ with the following properties 
		\begin{enumerate}[(i)]
			\item anonymous within clusters 
			\item has welfare at least as high as the first-best welfare \emph{in each cluster }
		\end{enumerate}
		is Pareto-improving.
	\end{lemma}
	
	Before we prove the lemma, we note that Pareto improvement of the scheme follows immediately given that it is anonymous within clusters by construction, and statewise dominates the welfare of the first-best scheme within each cluster (i.e., having agents team up according to the true ordering within each cluster but not across clusters). \\
			
	\begin{proofof}{Lemma~\ref{lem:welf-opt}}
			{We prove this by contradiction. Suppose $\mathcal{S}$ satisfies properties $(i)$ and $(ii)$, but is not Pareto-improving. Then, there exists an agent $i^\star$ for which $\mathcal{U}_{i^\star}(\mathcal{S}) < \mathcal{U}_{i^\star}(\emptyset)$. Let $k$ denote the cluster of this agent, $n_{k}$ the number of agents in $k$, and $\mathcal{S}_k^\star$ the signaling scheme that implements the first-best welfare in cluster $k$.} 
			
			{Since $\mathcal{S}$ is anonymous, for all $j \in k$, $\mathcal{U}_j(\mathcal{S}) = \mathcal{U}_{i^\star}(\mathcal{S})$. Thus, the welfare of cluster $k$ is 
				\begin{align*}
				\sum_{i=1}^{n_{k}} \mathcal{U}_i(\mathcal{S}) &= n_k \, \mathcal{U}_{i^\star}(\mathcal{S}) < n_k \, \mathcal{U}_{i^\star}(\emptyset) = \sum_{i=1}^{n_k} \mathcal{U}_i(\emptyset) \leq \sum_{i=1}^{n_k} \mathcal{U}_i(\mathcal{S}_k^{\star})
				\end{align*}
				where the final inequality follows from the fact the $\mathcal{S}_k^\star$ implements the first-best welfare, so consequently must have welfare at least as high as any other within-cluster signaling scheme, including the no-information scheme.}
			
			{The chain of inequalities implies that $\sum_{i=1}^{n_k} \mathcal{U}_i(\mathcal{S}) < \sum_{i=1}^{n_k} \mathcal{U}_i(\mathcal{S}_k^{\star})$, a contradiction, since we assumed that $\mathcal{S}$ improves upon first-best in each cluster.}
	\end{proofof}
	
	\end{proof}}

\subsubsection{Strictly concave utilities.}

{Recall the \emph{Cluster First Best} scheme, defined in Section~\ref{ssec:self-ag-concave}, which maximally pairs type 1 and type 0 agents for teams of size 2. Consider an analogous construction for teams of size $a > 2$. We also denote this scheme \emph{Cluster First Best}.}

{\begin{theorem}\label{thm:self-agnostic-biggerteams}
	Given $K \geq 1$ clusters with $n_1,\ldots,n_K$ agents such that $\sqrt{\frac{\ln n_1}{n_1}} \leq |p_1-1/2|,\ldots,\sqrt{\frac{\ln n_K}{n_K}} \leq |p_K-1/2|$, concave utilities and team size $a > 2$, Cluster First Best achieves $o(1)$ regret.
\end{theorem}}

{We only provide a proof sketch of the theorem, given how naturally the proofs of Theorem~\ref{thm:2-clusters} and Corollary~\ref{cor:k-clusters} extend.

\begin{proof}[Proof sketch]
	We first show that \emph{Cluster First Best} is Pareto-improving for teams of size $a > 2$. This fact follows immediately from Lemma~\ref{lem:welf-opt}.
				
		{To argue that \emph{Cluster First Best} is persuasive for teams of arbitrary size, we note that the proof of persuasiveness for teams of two did not crucially rely on the size of the teams. Persuasiveness, then, is immediate from the $a = 2$ proof.}
		
		It remains for us to show that \emph{Cluster First Best} is asymptotically optimal for $a > 2$. 
		An analogous argument applies if we choose to maximize the above function: namely, taking the dual and restricting to a typical set, proving via induction that we cannot improve on \emph{Cluster First Best} via inter-cluster swaps, and then showing via concentration arguments that we lose at most $o(1)$ when expanding to the space of all realizations. We leave the extension of this proof as an exercise to the reader.

		The above proof sketch is for $K = 2$ clusters and teams of arbitrary size. For $K > 2$ clusters and teams of arbitrary size, Corollary~\ref{cor:k-clusters} extends fairly naturally as well. In particular, the proof proceeds by induction on the number of clusters, and showing that if \emph{Cluster First Best} is not dual feasible (over the typical set) for $K+1$ clusters, then it is not dual feasible (over the typical set) for $K$ clusters. Standard concentration arguments give the result.
\end{proof}}

\subsection{Self-aware agents}

{For teams of arbitrary size, the impossibility in the self-aware setting extends, with an additional assumption on the structure of the utility function $u$.}

{\begin{definition}[Discrete regularity]
\begin{enumerate}[(i)]
\item A \emph{concave} utility function $u$ is said to be \emph{discrete regular} if the following holds:
$$ k_1 u(k_1) - (k_1-1)u(k_1-1) > (k_2 + 1)u(k_2 + 1) - k_2u(k_2), \quad \forall \, k_1 > k_2 \geq 1 $$
\item A \emph{convex} utility function $u$ is said to be \emph{discrete regular} if the following holds for all $a \geq 2$:
$$ \left(a-(k_1+1)\right) u(k_1+1) - \left(a-k_1\right) u(k_1) < \left(a-k_2\right)u(k_2) - \left(a-(k_2 - 1)\right)u(k_2 - 1), \quad \forall \, k_1 \geq k_2 \geq 1 $$
\end{enumerate}
\end{definition}
}

{At a high level, discrete regularity enforces that the curvature of the utility functions not be too extreme. Moreover, in the case of strictly concave utility functions, this notion of discrete regularity is analogous to that of regularity in classical mechanism design~\cite{hartline2013mechanism}. }

{\begin{theorem}
In the setting where agents are self-aware and symmetric, any \emph{discrete regular} convex or concave utility function $u$ and team size $a > 2$, no {Pareto improving} signaling policy can perform better than random. Moreover
\begin{enumerate}[(i)]
\item For convex utilities, the agents blocking any other signaling scheme are type 0 agents.
\item For strictly concave utilities, the agents blocking any other signaling scheme are type 1 agents.
\end{enumerate}
\end{theorem}

\begin{proof}
We prove the theorem for strictly concave utilities; the proof for convex utilities is analogous. We will show that, given a random (myopic) configuration of teams, any local improvement to welfare necessarily results in a decrease in utility for type 1 agents.

Consider any realization $\theta$, and let $h_t(\theta), \ell_t(\theta)$ denote the number of type 1 and type 0 agents in team $t$. For ease of notation, we suppress the dependence on $\theta$ for the remainder of the proof.

Suppose there exist teams $t, t'$ such that $u(h_t) + u(h_{t'}) < u\left(h_t + 1\right) + u\left(h_{t'}-1\right)$, with $h_t < h_{t'}$. By concavity of $u$, it is welfare optimal to switch to the team configuration with $\tilde{h}_{t} = h_{t} + 1, \tilde{h}_{t'} = h_{t'} - 1$. Consider now the aggregate change in utility of all type 1 agents.
\begin{align*}
\Delta\mathcal{U} &= \left(\left(h_{t}+1\right)u(h_t + 1) + \left(h_{t'} - 1\right)u(h_{t'} - 1)\right) - \left(h_t u(h_t) + h_{t'} u(h_{t'})\right)  \\ 
&< 0 \text{ by discrete regularity}
\end{align*}
Due to the fact that agents are symmetric, this immediately implies that the utility of each type 1 agent strictly decreases from a swap to a welfare-improving configuration.
\end{proof}}

\section{Discussion}
\label{sec:extensions}

{Our work uncovered the key role of awareness in the design of signaling mechanisms for team formation by presenting two extreme cases -- that of agents who don't know their own types, and that of agents who have full knowledge of their own types. A natural next step would be to ask how \emph{partial awareness} impacts the design of our signaling schemes. One potential model for partial awareness would be that agents are drawn from a known mixture distribution, and have private information regarding the sub-distribution from which they are drawn, but not the realization of their type. The interpretation of such a knowledge structure would be that each agent knows the cluster to which she belongs, but not the clusters to which other agents do. For strictly concave utilities, it would seem as though the principal can achieve far better than \emph{Cluster First Best}, even asymptotically. We leave the development and analysis of such a model as future work.}

Although our results consider the setting in which agents' types are binary, this is not critical for our techniques, which therefore should extend to more general discrete distributions (e.g., multinomial). The main additional ingredient that we need is to define analogs for the agents' utility function (and in particular, appropriate notions of convexity) in settings with more types. For example, in the self-agnostic case with one cluster, with linear and symmetric agent utilities and appropriate notions of convexity, a nearly identical construction would give us that \emph{First Best} is achievable by the principal. The setting with many clusters requires additional assumptions on the distributions from which each cluster is drawn. Namely, for binary types, our results relied on the fact that $p_1 > p_2$. Intuitively, however, first-order stochastic dominance should be sufficient for similar results to hold. Such a generalization however would require extensive additional notation for handling the larger type-spaces, and hence we leave this for future work.

Another related setting is that of information signalling in \emph{bipartite matching}. One motivating example is that of matching riders and drivers on ridesharing platforms. On these platforms, both riders and drivers have ratings. Companies such as Lyft and Uber have an incentive to want diverse matches, since matches of poorly rated riders to poorly rated drivers result in negative experiences for all users, and are likely to create self-reinforcing cycles of poor ratings.  If there was a way for the platform to obtain private information on each side (say, selectively display some of the ratings), then there is hope for socially optimal matches to be induced via appropriate signaling mechanisms. 

In more detail, in the bipartite matching problem, the one-cluster analog would be each side of the platform drawn from a $Ber(p_{left}), Ber(p_{right})$ distribution. Again, an analogous construction using the `complement' of a realized type profile imply similar results. Further, the case with multiple clusters on each side lends itself to the intuitive result that the platform cannot benefit from having, for example, the highest ranked cluster on the left match with any other cluster than the highest ranked cluster on the right. We believe many of our techniques and ideas should prove useful in this setting, and leave it as a promising avenue for future work.

\newpage

\appendix

\section{Proof of Theorem~\ref{thm:fb-symm}}\label{sec:app-thm2-proof}
\begin{proof}
	
	We first show that solving the relaxed LP (\ref{relaxed-lp}) allows the principal to achieve the minimum number of matches between two type 0 agents, for each realization $\theta$.
	
	\begin{lemma}\label{lem:relaxed-sol}
		$$ \widehat{OPT} = \expectation{\min_{\sigma}m_{00}(\sigma(\theta))} $$
	\end{lemma}
	
	\begin{proof}

	The proof is constructive. Let $\I(\theta), \J(\theta)$ denote the number of agents of type 0 and type 1, respectively. Clearly, we have that
	\begin{align*}
	\I(\theta) = n-\J(\theta).
	\end{align*}
	\\
	\underline{Case 1:} $\J(\theta) \in \{0,n\}$
	
	When $\J(\theta) = 0$, for all $\sigma\in\bSigma $, $m_{00}(\sigma(\theta)) = n/2$. Thus, the principal is indifferent among orderings, and can choose an arbitrary ordering $\tilde{\sigma}$ and announce it with probability 1 ($x_{\theta,\tilde{\sigma}} = 1$).
	
	Similarly, when $\J(\theta) = n$, $m_{00}(\sigma(\theta)) = 0$ for all $\sigma\in\bSigma $. As in the case above, the principal can choose an arbitrary ordering to announce with probability 1.\\
	
	\noindent
	\underline{Case 2:} $\J(\theta) \in \{1,n-1\}$
	
	When $\J(\theta) = 1$, $m_{00}(\sigma(\theta)) = n/2-1$ for all $\sigma$. Additionally, when $\J(\theta) = n-1$, $m_{00}(\sigma(\theta)) = 0$ for all $\sigma$. As in Case 1, the principal is indifferent between orderings.
	
	Let $\sigma^{\text{truth}}(\theta)$ denote any ordering that correctly orders agents for a realization $\theta$. We henceforth omit the dependence on $\theta$ when clear from context.  Pick an arbitrary $\sigma^{\text{truth}}$ and set $x_{\theta,\sigma^{\text{truth}}} = 1$.\\
	
	\noindent
	\underline{Case 3:} $1 < \J(\theta) < n-1$
	
	In the same spirit as Example~\ref{ex:fb}, our goal will be to find a realization $\bar\theta$ with which to pair $\theta$ such that $\sigma^*(\theta) \in \arg\min\limits_{\sigma}m_{00}(\sigma(\theta))$, and $\sigma^*(\bar\theta) \in \arg\min\limits_{\sigma}m_{00}(\sigma(\bar\theta))$. 
	
	Suppose $\J(\theta) > \I(\theta)$, and let $e(\theta) \triangleq \J(\theta)-\I(\theta)$ denote the excess of type 1 agents over type 0 agents. (An analogous construction can be shown to be feasible for $\J(\theta) \leq \I(\theta)$.) {For ease of notation, we suppress the dependence of $e$ on $\theta$.} Let ordering $\sigma^*$ be any permutation of $\theta$ such that there are as many consecutive type 1 and type 0 agents as possible under this ordering. That is, for such $\sigma^*$, ${\sigma^*_1}(\theta) = 1, {\sigma^*_2}(\theta) = 1, \ldots, {\sigma^*_e}(\theta) = 1$, and ${\sigma^*_{e+1}}(\theta) = 1, {\sigma^*_{e+2}}(\theta) = 0, \ldots, {\sigma^*}_{n-1}(\theta) = 1, {\sigma^*_n}(\theta) = 0$. Clearly, ${\sigma^*} \in \arg\min\limits_{\sigma}m_{00}(\sigma(\theta))$.

	Let $\bar\theta$ be the realization with $\J(\bar\theta) = \J(\theta)$, and such that, for the same $\sigma^*$ defined above, ${\sigma^*_1}(\bar\theta) = 1, {\sigma^*_2}(\bar\theta) = 1, \ldots, {\sigma^*_e}(\bar\theta) = 1$, and ${\sigma^*_{e+1}}(\bar\theta) = 0, {\sigma^*_{e+2}}(\bar\theta) = 1, \ldots, {\sigma^*_{n-1}}(\bar\theta) = 0, {\sigma^*_n}(\bar\theta) = 1$. This construction also makes it clear that ${\sigma^*} \in \arg\min\limits_{\sigma}m_{00}(\sigma(\bar\theta))$.
	
	Let $x_{\theta,{\sigma^*}} = x_{\bar\theta,{\sigma^*}} = 1$. We show that this construction satisfies the set of `persuasive' constraints of (\ref{relaxed-lp}), namely:
	\begin{align*}
	&\sum_{\theta\in\bTheta }\lambda(\theta)x_{\theta,{\sigma^*}}\sig{i}^*(\theta) \geq \sum_{\theta\in\bTheta }\lambda(\theta)x_{\theta,{\sigma^*}}\sig{i+1}^*(\theta) &   \forall \; i \in [n-1]\\
	\iff &\sum_{\theta\in\bTheta }\lambda(\theta)x_{\theta,{\sigma^*}}\bar{v}(\sig{i}^*(\theta)) \geq 0 &   \forall \; i \in [n-1],
	\end{align*}
	where $\bar{v}(\sig{i}^*(\theta)) \triangleq \sig{i}^*(\theta)-\sig{i+1}^*(\theta)$.

	For $1 \leq i \leq {e}$, all entries of ${\sigma^*}(\theta)$ and ${\sigma^*}(\bar\theta)$ are 1. Thus, for $1 \leq i \leq {e}-1$
	\begin{align*}
	\bar{v}(\sig{i}^*(\theta)) = \bar{v}(\sig{i}^*(\bar\theta)) = 0
	\implies \lambda(\theta)x_{\theta,{\sigma^*}}\bar{v}(\sig{i}^*(\theta)) + \lambda(\bar\theta)x_{\bar\theta,{\sigma^*}}\bar{v}(\sig{i}^*(\bar\theta)) = 0.
	\end{align*}
	
	As noted in the example, since $\J(\bar\theta) = \J(\theta)$, and all agent types are drawn independently from the same $Ber(p)$ distribution, we have $\lambda(\bar\theta) = \lambda(\theta)$. Additionally, for $i \geq {e}+1$, our construction is such that $\bar{v}(\sig{i}^*(\theta)) = -\bar{v}(\sig{i}^*(\bar\theta))$. These two facts together give us that, for $i \geq {e}+1$
	\begin{align*}
	\lambda(\theta)x_{\theta,{\sigma^*}}\bar{v}(\sig{i}^*(\theta)) + \lambda(\bar\theta)x_{\bar\theta,{\sigma^*}}\bar{v}(\sig{i}^*(\bar\theta)) = \lambda(\theta)(\bar{v}(\sig{i}^*(\theta))-\bar{v}(\sig{i}^*(\theta))) = 0.
	\end{align*}
	
	It remains for us to show that the constraint is satisfied for $i = {e}$. By construction, $\bar{v}(\sig{e}^*(\theta)) = 0$, and $\bar{v}(\sig{e}^*(\bar\theta)) = 1$. Thus:
	\begin{align*}
	\lambda(\theta)x_{\theta,{\sigma^*}}\bar{v}(\sig{i}^*(\theta)) + \lambda(\bar\theta)x_{\bar\theta,{\sigma^*}}\bar{v}(\sig{i}^*(\bar\theta)) &= \lambda(\theta) > 0.
	\end{align*} 
	
	\end{proof}
	
	We have shown that we can achieve first-best in our relaxed problem. Although the construction that we have presented is not necessarily symmetric (i.e., not all type 1 agents have the same probability of being matched with a type 0 agent), we can randomize this scheme over realizations, since agents themselves are symmetric. This allows the \emph{First Best} scheme to be both symmetric  and feasible in (\ref{relaxed-lp}). Thus, henceforth we use the term `\emph{First Best}' to refer to the symmetric version of the above signaling scheme.
	
	It remains for us to show that \emph{First Best} is also {Pareto improving}. That is, each agent's expected utility under the signaling mechanism is at least as large as her expected utility if she were to act myopically (which, in the one-cluster case, would be simply to pair up with someone randomly).
	
	\begin{lemma}\label{lem:fb-ir}
		The \emph{First Best} signaling scheme is {Pareto improving}. Thus, it is optimal for all agents to follow the principal's signal.
	\end{lemma}
	
	\begin{proof}

	Let $\mathcal{S}^{FB}$ denote the \emph{First Best} scheme. {Under $\mathcal{S}^{FB}$, if there are more type 0 agents than type 1 agents, a type 1 agent will \emph{always} be matched with a type 0 agent; on the other hand, if there are more type 1 agents than type 0 agents, a randomly chosen type 1 agent will be matched with another randomly chosen type 1 agent. Similarly for type 0 agents: if there are more type 1 agents than type 0 agents, a type 0 agent will \emph{always} be matched with a type 1 agent; otherwise randomly chosen type 0 agents will match together.}

	Thus, the expected utility of agent $i$ $\mathcal{U}_i(\mathcal{S}^{FB})$ is given by:
	\begin{align*}
	\mathcal{U}_i(\mathcal{S}^{FB}) &= \EE\Bigg[\theta_i\left(\min\left\{1,\frac{\I(\theta)}{\J(\theta)}\right\}u(1,0)+\left(1-\min\left\{1,\frac{\I(\theta)}{\J(\theta)}\right\}\right)u(1,1)\right)\\
	&\quad+(1-\theta_i)\left(\min\left\{1,\frac{\J(\theta)}{\I(\theta)}\right\}u(1,0)+\left(1-\min\left\{1,\frac{\J(\theta)}{\I(\theta)}\right\}\right)u(0,0)\right)\Bigg]\\
	&= \EE\Bigg[\theta_i\left(\min\left\{1,\frac{\I(\theta)}{\J(\theta)}\right\}\left(u(1,0)-u(1,1)\right)+u(1,1)\right)\\
	&\quad+(1-\theta_i)\left(\min\left\{1,\frac{\J(\theta)}{\I(\theta)}\right\}\left(u(1,0)-u(0,0)\right)+u(0,0)\right)\Bigg]\\
	&= pu(1,1)+(1-p)u(0,0)+\left(u(1,0)-u(1,1)\right)\EE\Bigg[\theta_i\min\left\{1,\frac{\I(\theta)}{\J(\theta)}\right\}\Bigg]\\
	&\quad +\left(u(1,0)-u(0,0)\right) \EE\Bigg[(1-\theta_i)\min\left\{1,\frac{\J(\theta)}{\I(\theta)}\right\}\Bigg]\\
	&= pu(1,1)+(1-p)u(0,0) + \frac{1}{n}\left(2u(1,0)-u(1,1)-u(0,0)\right)\EE\Big[\min\left\{\I(\theta,\J{(\theta)})\right\}\Big],
	\end{align*}
	where the last equality results from the easy-to-show fact that in the symmetric setting $$ \EE\Bigg[\theta_i\min\left\{1,\frac{\I(\theta)}{\J(\theta)}\right\}\Bigg] = \EE\Bigg[(1-\theta_i)\min\left\{1,\frac{\J(\theta)}{\I(\theta)}\right\}\Bigg] = \frac{1}{n}\EE\Big[\min\{\I(\theta),\J(\theta)\}\Big].$$
	
	We present an alternate expression for the expected utility of agent $i$ under her myopic strategy (random matching), in the symmetric case. We defer the easy proof of this proposition to Appendix~\ref{sec:aux-proofs}.

\begin{prop}\label{prop:baseline-util}
	In the self-agnostic setting with i.i.d. types, agent $i$'s baseline utility is given by
	\begin{align*}
	\mathcal{U}_i(\emptyset) = pu(1,1) + (1-p)u(0,0) + p(1-p)\Big(2u(1,0)-u(1,1)-u(0,0)\Big).
	\end{align*}
\end{prop}
	
	Thus, by concavity of $u$, it suffices for us to show that
	\begin{align*}
	\frac{1}{n}\EE\left[\min\{\I(\theta),\J(\theta)\}\right] \geq p(1-p) \qquad \forall \; n \geq 2, p\in[0,1].
	\end{align*}
	
	We have:
	\begin{align*}
	\frac{1}{n}\EE\left[\min\{\I(\theta),\J(\theta)\}\right] &= \frac{1}{n}\sum_{k=1}^n k {n\choose k}\left[p^k(1-p)^{n-k}+(1-p)^kp^{n-k}\right]\\
	&= p(1-p)\sum_{k=1}^n {n-1\choose k-1} \left[p^{k-1}(1-p)^{n-k-1}+(1-p)^{k-1} p^{n-k-1}\right]\\
	&= p(1-p)\sum_{k=0}^{n-1} {n-1\choose k} \left[p^k(1-p)^{n-k-2}+(1-p)^kp^{n-k-2}\right]\\
	&\geq p(1-p) \sum_{k=0}^{n-1} {n-1\choose k} \left[p^k(1-p)^{n-k-1}+(1-p)^kp^{n-k-1}\right]\\
	&= p(1-p)\left[\sum_{k=0}^{n-1}{n-1\choose k}p^k(1-p)^{n-1-k}+\sum_{k=n}^{n-1}{n-1\choose k}p^k(1-p)^{n-k-1}\right]\\
	&= p(1-p).
	\end{align*}
	\end{proof} 
Theorem~\ref{thm:fb-symm} follows from Lemmas~\ref{lem:relaxed-sol} and \ref{lem:fb-ir}, since we have shown that the optimal scheme that achieves a lower bound on (\ref{opt}) is also feasible.
\end{proof}

	\section{Proofs of Theorem~\ref{thm:2-clusters} and Corollary~\ref{cor:k-clusters}}\label{ssec:thm-3-proof}

\begin{proofof}{Theorem~\ref{thm:2-clusters}}		
		Let $\typical{2}$ be the set of all realizations $\theta \in \bTheta $ such that the number of type 1 individuals in each cluster is within an $\epsilon_1n_1,\epsilon_2n_2)$ neighborhood of $n_1p_1,n_2p_2$, respectively. 
		\begin{align*}
		\typical{2} \triangleq \left\{\theta: |\J_1(\theta)-n_1p_1|\leq \epsilon_1n_1,|\J_2(\theta)-n_2p_2|\leq \epsilon_2n_2 \right\}.
		\end{align*}
		
		We will require $\epsilon_1,\epsilon_2$ to be such that $\typical{2}$ is entirely contained in the same orthant as $(n_1p_1,n_2p_2)$. That is, $\epsilon_1 \leq |p_1-1/2|, \epsilon_2 \leq |p_2-1/2|$. 
		Additionally, we define $\bar{A}_{\varepsilon}^{(n,2)}$ to be the complement of $\typical{2}$.
		
		We begin by finding an upper bound on $OPT$. Recall that our objective is with respect to $m_{00}$, the number of `0-0' matches. We abuse notation and use $FB_C(\theta)$ to denote the number of `0-0' matches created by the \emph{Cluster First Best} scheme (optimally matching agents in each cluster), for a fixed realization $\theta$.
		
		\begin{lemma}\label{lem:ub-2-clusters}
			\begin{align}
			OPT \leq \sum_{\theta\in \typical{2}}\lambda(\theta)FB_C(\theta)+o_{n_1+n_2}(1).\label{eq:ub-fb-c}\tag{UB}
			\end{align}
		\end{lemma}
		
		\begin{proof}
			We use the Cluster First Best scheme, which we know to be feasible, for our upper bound.
			\begin{align}
			OPT &\leq \sum_{\theta\in\bTheta }\lambda(\theta)FB_C(\theta) \notag\\
			&=\sum_{\theta\in \typical{2}} \lambda(\theta)FB_C(\theta)+\sum_{\theta\in {\bartypical{2}}} \lambda(\theta)FB_C(\theta) \notag \\
			&\leq \sum_{\theta\in \typical{2}} \lambda(\theta)FB_C(\theta)+n\sum_{\theta\in {\bartypical{2}}} \lambda(\theta) \label{eq:at-most-n}\\
			&\leq \sum_{\theta\in \typical{2}} \lambda(\theta)FB_C(\theta)+n\left(e^{-2\epsilon_1^2n_1}+e^{-2\epsilon_2^2n_2}\right) \label{eq:hoeffding}
			\end{align}
			where inequality~\eqref{eq:at-most-n} uses the crude upper bound of $FB_C(\theta) \leq n$, and inequality~\ref{eq:hoeffding} follows from Hoeffding's inequality~\cite{hoeffding1963probability}.
			
			Let $\epsilon_1 =\sqrt{\frac{\ln n_1}{n_1}}, \epsilon_2 = \sqrt{\frac{\ln n_2}{n_2}}$. Since $\sqrt{\frac{\ln n}{n}}$ is monotonically decreasing for $n\geq 2$, there exists $n_0$ such that for all $n_1 \geq n_0, n_2 \geq n_0, \epsilon_1 \leq |p_1-1/2|, \epsilon_2 \leq |p_2-1/2|$, which satisfies our requirement of staying within the same region as $(np_1,np_2)$. Thus, we obtain the bound:
			\begin{align*}
			OPT &\leq \sum_{\theta\in \typical{2}} \lambda(\theta)FB_C(\theta)+2n\cdot\frac{1}{n_1^2+n_2^2}  \\
			&= \sum_{\theta\in \typical{2}} \lambda(\theta)FB_C(\theta)+\frac{2(n_1+n_2)}{n_1^2+n_2^2}  \\
			&= \sum_{\theta\in \typical{2}} \lambda(\theta)FB_C(\theta) + o_{n_1+n_2}(1).
			\end{align*} 
		\end{proof}
		
		We treat the analysis of the lower bound separately, depending on the orthant in which $(p_1,p_2)$ resides:
		1. $p_1 > p_2 > 1/2$\quad,\quad
		2. $p_2 < p_1 < 1/2$\quad,\quad
		3. $p_2 < 1/2 < p_1$.
		
		The fact that $p_1 > p_2$ is without loss of generality, since we can relabel the clusters otherwise. Additionally, recall that $p_1 = p_2$ is the scenario in which we only have one cluster, for which we have shown that we can achieve first-best. 
		
		Let $FB(\theta) =  \min_{\sigma} m_{00}(\sigma(\theta))$.That is, $FB(\theta)$ is the minimum number of `0-0' matches the principal can achieve without any Pareto improvement or stability constraints.\\
		
		\noindent
		\underline{Cases 1 and 2:}
		
		In these orthants, we have the following fact. 
		
		\begin{prop}\label{prop:fb=fb-c}
			For all $\theta \in \typical{2}$, and settings where $p_1 > p_2 \geq 1/2$ or $1/2 \leq p_2 < p_1$, $FB_C(\theta) = FB(\theta)$.
		\end{prop}
		
		\begin{proof}
			Since the Cluster First Best scheme matches as many type 1 agents to type 0 agents as possible in each cluster, the number of `0-0' matches in each cluster is simply half of the remaining (or excess) type 0 agents. Thus, we have:
			\begin{align*}
			FB_C(\theta) &= \left(\frac{\I_1(\theta)-\J_1(\theta)}{2}\right)^++\left(\frac{\I_2(\theta)-\J_2(\theta)}{2}\right)^+\\
			&=\begin{cases}
			0 & \mbox{ if } p_1 > p_2 \geq 1/2\\
			\frac{\I_1(\theta)-\J_1(\theta)}{2}+\frac{\I_2(\theta)-\J_2(\theta)}{2} & \mbox{ if } p_2 < p_1 < 1/2
			\end{cases}
			\\
			&=\left(\frac{(\I_1(\theta)+\I_2(\theta))-(\J_1(\theta)+\J_2(\theta))}{2}\right)^+\\
			&= FB(\theta).
			\end{align*}
		\end{proof}
		
		Using Proposition~\ref{prop:fb=fb-c}, we can re-write the upper bound (\ref{eq:ub-fb-c}) as follows:
		\begin{align}\label{ub1}
		OPT \leq \sum_{\theta\in \typical{2}}\lambda(\theta)FB(\theta)+o_{n_1+n_2}(1).\tag{UB1}
		\end{align}
		
		The proof of the lower bound for Cases 1 and 2 is straightforward. We use \emph{First Best} as a crude lower bound on $OPT$. 
		\begin{align}
		OPT &\geq \sum_{\theta\in\bTheta }\lambda(\theta)FB(\theta)\nonumber\\
		&\geq \sum_{\theta\in \typical{2}} \lambda(\theta)FB(\theta)\nonumber\\
		&= \sum_{\theta\in \typical{2}} \lambda(\theta)FB_C(\theta)\label{lb1}\tag{LB1}.
		\end{align}
		
		Putting the upper and lower bounds (\ref{ub1}) and (\ref{lb1}) together, we get that Cluster First Best achieves $o_{n_1+n_2}(1)$ regret for Cases 1 and 2.\\
		
		\noindent
		\underline{Case 3:} We consider the case where $p_1 > 1/2 > p_2$. {Recall that we've defined the difference between the types of two consecutive agents in an ordering as $\bar{v}({\sigma}_i(\theta)) \triangleq {\sigma}_i(\theta)-{\sigma}_{i+1}(\theta)$}.
		
		We obtain a lower bound on $OPT$ via the dual, given by:
		\begin{align}\label{dual}
		\begin{aligned}
		\max_{z_{\theta} \text{ free},y_{\sigma i} \geq 0}& &\sum_{\theta}z_{\theta}&\\
		\text{subject to }& &z_{\theta}+\sum_{i=1}^{n-1}\lambda(\theta)\bar{v}(\sig{i}(\theta))y_{\sigma i} \leq \lambda(\theta)m_{00}(\sigma(\theta)) & \quad \forall \; \theta\in\bTheta ,\sigma\in\bSigma &
		\end{aligned}
		\end{align}
		
		Consider the restricted space of realizations $ \typical{2}$. In the case where $p_1 > 1/2 > p_2$, by construction $\typical{2} = \{\theta: \J_1(\theta) > n_1/2, \J_2(\theta) < n_2/2\}$. This set is of interest because it corresponds to the realizations for which Cluster First Best fails, and more sophisticated policies would succeed. Note that, for $\theta \in \typical{2}$, $FB_C(\theta) = \frac{\I_2(\theta)-\J_2(\theta)}{2}$. If we modify the dual to be over this restricted space $\typical{2}$, we have the following result.
		
		\begin{lemma}\label{lem:fixed-num}
			For all $\sigma\in \bSigma , i\in [n-1]$, there exists $y_{\sigma i}$ such that $z_{\theta} = \lambda(\theta)\cdot\left(\frac{\I_2(\theta)-\J_2(\theta)}{2}\right) = \lambda(\theta) FB_C(\theta)$, and $(z_{\theta}, y_{\sigma_i})$ form a feasible set of solutions in the restricted space of realizations $ \typical{2}$.
		\end{lemma}

		To show this, we need to show that there exists $y_{\sigma i} \geq 0, \forall\; \sigma,i$, such that:
		\begin{align}
		\sum_{i=1}^{n-1}\lambda(\theta)\bar{v}(\sig{i}(\theta))y_{\sigma i} &\leq \lambda(\theta)m_{00}(\sigma(\theta))-  \lambda(\theta)\cdot\frac{\I_2(\theta)-\J_2(\theta)}{2} \quad \forall \; \theta \in \typical{2} \nonumber \\
		\iff \sum_{i=1}^{n-1}\bar{v}(\sig{i}(\theta))y_{\sigma i} &\leq m_{00}(\sigma(\theta)) -  \frac{\I_2(\theta)-\J_2(\theta)}{2} \quad \forall \; \theta \in \typical{2}.\label{eq:dual-feas}
		\end{align}
		
		We will prove something stronger -- namely, we will find a set of \emph{binary} $y_{\sigma i}$ such that Equation (\ref{eq:dual-feas}) holds. This problem, then, reduces to finding an index set $\mathcal{I}_{\sigma}$ for each permutation $\sigma$ such that
		\begin{align}
		m_{00}(\sigma(\theta)) - \sum_{i\in\mathcal{I}_{\sigma}}\bar{v}(\sig{i}(\theta))&\geq \frac{\I_2(\theta)-\J_2(\theta)}{2} \quad \forall \; \theta\in \typical{2}\label{eq:to-show}.
		\end{align}
		
		We approach finding this set $\mathcal{I}_{\sigma}$ through the following key observation: any permutation $\sigma$ can be produced through a sequence of swaps of pairs of agents. Let the `identity' permutation be the ordering $\sigma^0 = (1,2,\ldots,n_1,n_1+1,\ldots,n_1+n_2)$. Without loss of generality, we can restrict to reasoning about swaps of two agents with respect to the identity permutation, since agents within clusters can be arbitrarily re-labeled.
		
		At a high level, we seek to show that the principal does not benefit from swapping individual across clusters. We ignore the effect of swapping individuals within clusters, since we already know that Cluster First Best does this in the optimal way, and again, agents within clusters are i.i.d. and can  be relabeled. {A similar line of reasoning tells us that, once we have created an inter-cluster match, we do not benefit from swapping an agent back to her original cluster, or matching her with another agent of the same cluster as her current match.} Thus, it is sufficient to reason about the swaps that have created inter-cluster matches.
		
		Let $N$ be the number of swaps that created inter-cluster matches, and $\sigma^{N}$ be the permutation which created the $N$th inter-cluster match. We denote the identities of the  $N$th pair of swapped individuals (after $\swapsig{N}$) to be $\indexswapsig{i_1}{N}, \indexswapsig{i_2}{N}$, respectively, where the subscript under the $i$  represents the cluster that each agent finds herself in \emph{after} the swap. Additionally, let $\indexswapsig{m(\swappedind{1})}{N},\indexswapsig{m(\swappedind{2})}{N}$ denote the respective identities of the partners of the two swapped agents from the $N$th swap. Finally, without loss of generality we assume that {swaps are made in a {nested} fashion.} This is simply done for ease of notation, as the individuals are in each cluster are symmetric and can be relabeled. Let $N_{\max}$ denote the maximum number of swaps that can create inter-cluster matches with respect to the identity permutation.
		
		{\begin{example}
				Suppose $n=4$ with agents A, B, C, and D. Additionally, suppose A and B are in Cluster 1, and C and D are in Cluster 2, and we've performed $N = 1$ inter-cluster swaps, with $\swapsig{1} = (A \succcurlyeq C \succcurlyeq B \succcurlyeq D)$ and $\theta = (1,1,0,0)$. Then, $\swapsig{1}(\theta) = (1, 0, 1, 0)$.  Further, 
				\begin{itemize}
					\item \emph{Post-swap}: $\indexswapsig{\swappedind{1}}{1} = C$ $\indexswapsig{\swappedind{2}}{1} = B$; and $\indexswapsig{\swappedind{1}}{1}(\theta) = 0$ and $\indexswapsig{\swappedind{2}}{1}(\theta) = 1$; $\indexswapsig{m(\swappedind{1})}{1} = A$, $\indexswapsig{m(\swappedind{2})}{1} = D$; and $\indexswapsig{m(\swappedind{1})}{1}(\theta) = 1$, $\indexswapsig{m(\swappedind{2})}{1}(\theta) = 0$
					\item \emph{Pre-swap}: $\indexswapsig{\swappedind{1}}{0} = B$ $\indexswapsig{\swappedind{2}}{0} = C$; and $\indexswapsig{\swappedind{1}}{0}(\theta) = 1$ and $\indexswapsig{\swappedind{2}}{0}(\theta) = 0$; $\indexswapsig{m(\swappedind{1})}{0} = A$, $\indexswapsig{m(\swappedind{2})}{0} = D$; and $\indexswapsig{m(\swappedind{1})}{0}(\theta) = 1$, $\indexswapsig{m(\swappedind{2})}{0}(\theta) = 0$.
				\end{itemize}
				Note in particular that the identities (and consequently, types) of the partners pre- and post-swap have not changed.
			\end{example}}
			
			Recall that for fixed $\sigma, \theta, i \in [n-1]$, we define $\bar{v}(\sig{i}(\theta))$ to be the difference between two consecutive agent types under ordering $\sigma$. That is, $\bar{v}(\sig{i}(\theta)) \triangleq \sig{i}(\theta) - \sig{i+1}(\theta)$. We define $\mathcal{I}_{\sigma^N}$ recursively, as follows. For all $0 \leq N \leq N_{\max}$, $N$ odd, $\indexset{\sigma}{N}$ is such that:
			\begin{align*}
			\sum_{i\in\indexset{\sigma}{N}}\bar{v}(\indexswapsig{i}{N}(\theta)) = \sum_{i\in\indexset{\sigma}{N-1}}\bar{v}(\indexswapsig{i}{N-1}(\theta))+\indexswapsig{\swappedind{1}}{N}(\theta)-\indexswapsig{\swappedind{2}}{N}(\theta).
			\end{align*}
			For all $2 \leq N \leq N_{\max}$, $N$ even, $\indexset{\sigma}{N}$ is such that:
			\begin{align*}
			\sum_{i\in\indexset{\sigma}{N}}\bar{v}(\indexswapsig{i}{N}(\theta)) = \sum_{i\in\indexset{\sigma}{N-1}}\bar{v}(\indexswapsig{i}{N-1}(\theta))-\indexswapsig{m(\swappedind{1})}{N}(\theta)+\indexswapsig{m(\swappedind{2})}{N}(\theta).
			\end{align*}
			For $N=0$, $\mathcal{I}_{\sigma^0} = \emptyset$, which implies  $\sum_{i\in\mathcal{I}_{\sigma^0}}\bar{v}(\indexswapsig{i}{0}(\theta)) = 0$.
			
			\begin{lemma}\label{lem:index-set}
				For the above choice of $\indexset{\sigma}{N}$, and for all $0 \leq N \leq N_{\max}$, 
				\begin{align}\label{eq:induc-ineq}
				m_{00}(\sigma(\theta)) - \sum_{i\in\indexset{\sigma}{N}}\bar{v}(\indexswapsig{i}{N}(\theta))\geq \frac{\I_2(\theta)-\J_2(\theta)}{2} \quad \forall \; \theta\in \typical{2}.
				\end{align}
			\end{lemma}
			
			{The proof of Lemma~\ref{lem:index-set} involves a lengthy induction and case-by-case analysis. We thus defer it to Appendix~\ref{ssex:app-b}.}
			
			Using Lemma~\ref{lem:fixed-num}, we now have the lower bound for the case where $p_1 > 1/2 > p_2$, in our restricted space of realizations.

			We construct a feasible solution to the dual using our $\{y_{\sigma i}\}$ construction. Recall our dual feasibility constraints:
			\begin{align*}
			&z_{\theta}+\sum_{i=1}^{n-1}\lambda(\theta)\bar{v}(\sig{i}(\theta))y_{\sigma i} \leq \lambda(\theta)m_{00}(\sigma(\theta)), \hspace{0.5cm}z_\theta \text{ free}, y_{\sigma i} \geq 0, \hspace{0.5cm} \forall \; \theta\in\bTheta ,\sigma\in\bSigma
			\end{align*}
			We already showed that, for all $\theta \in \typical{2}$ (i.e. such that $\J_1(\theta) > n_1/2,\J_2(\theta) < n_2/2$), $z_\theta = \lambda(\theta)\cdot\frac{\I_2-\J_2}{2}$ is feasible. Thus, it remains to consider $\theta \in \bartypical{2}$. For such $\theta$, given the $\left\{y_{\sigma i}\right\}$ that we constructed, we have:
			\begin{align*}
			z_\theta &= \min_{\sigma}\lambda(\theta)\left(m_{00}(\sigma(\theta))-\sum_{i=1}^{n-1}\bar{v}(\sig{i}(\theta))y_{\sigma i}\right)\\
			&=\min_{\sigma}\lambda(\theta)\left(m_{00}(\sigma(\theta)))-\sum_{i\in\mathcal{I}_\sigma}\bar{v}(\sig{i}(\theta))\right)\\
			&\geq -n\lambda(\theta),
			\end{align*}
			where we used the crude bounds $m_{00}(\sigma(\theta)) \geq 0$, and $ \sum_{i\in\mathcal{I}_{\sigma}}\bar{v}(\sig{i}(\theta)) \leq n$. Setting $z_\theta = -n\lambda(\theta)$ for $\theta \in \bartypical{2}$ is thus feasible.
			
			By weak duality:
			\begin{align}
			OPT &\geq \sum_{\theta\in \typical{2}}z_\theta -n\sum_{\theta\in \bartypical{2}}\lambda(\theta)\nonumber\\
			&\geq \sum_{\theta\in \typical{2}}z_\theta -\frac{2(n_1+n_2)}{n_1^2+n_2^2}\nonumber\\
			&= \sum_{\theta\in \typical{2}}\lambda(\theta)\cdot \frac{\I_2(\theta)-\J_2(\theta)}{2}-o_{n_1+n_2}(1)\nonumber\\
			&= \sum_{\theta\in \typical{2}}\lambda(\theta)FB_C(\theta)-o_{n_1+n_2}(1).\label{eq:lb-fb-c}\tag{LB2}
			\end{align}
			where the second inequality follows from Hoeffding's inequality, as was done to derive the upper bound.
			
			Putting the upper and lower bounds (\ref{eq:ub-fb-c}) and (\ref{eq:lb-fb-c}) together, we obtain our final result.
\end{proofof}\\

	\begin{proofof}{Corollary~\ref{cor:k-clusters}}
		Given $p_1,p_2,\ldots,p_K$, we define $A_{\varepsilon}^{(n,K)}$ analogously to the $K=2$ cluster case. That is:
		\begin{align}
		\typical{K} \triangleq \left\{\theta: |\J_1(\theta)-n_1p_1|\leq \epsilon_1n_1,|\J_2(\theta)-n_2p_2|\leq \epsilon_2n_2,\ldots,|\J_k(\theta)-n_Kp_K|\leq \epsilon_Kn_K \right\}.
		\end{align}
		As before, we will require that $\typical{K}$ be contained in the same orthant as $(n_1p_1,\ldots,n_Kp_K)$, i.e. $\epsilon_i \leq |p_i-1/2| \; \forall \; i\in[K]$.
		
		We only prove the corollary for the cases where there exists at least one cluster $i$ such that $p_i < 1/2$, and one cluster $j$ such that $p_j > 1/2$. Otherwise, the proof is identical to that of Proposition~\ref{prop:fb=fb-c}.
		
		Without loss of generality, we assume $p_1 > p_2 > \ldots > p_K$. Since we are only interested in the non-trivial cases, by the above we have that necessarily $p_1 > 1/2$.
		

		{	
			\begin{lemma}
				For all $\theta \in \typical{K}$, there exists $y_{\sigma i}$ such that $z_\theta = \lambda(\theta)FB_C(\theta)$ is dual feasible. 
			\end{lemma}
			
			\begin{proof}
				We prove the claim by induction.
				
				First recall that, for $K \geq 2$
				\begin{align}\label{eq:fbc-kclusters}
				FB_C(\theta) =  \left(\frac{\I_1(\theta)-\J_1(\theta)}{2}\right)^+ + \ldots + \left(\frac{\I_K(\theta)-\J_K(\theta)}{2}\right)^+ 
				\end{align} 
				
				For the rest of this proof, we suppress the dependence on $\theta$ for ease of notation.\\
				
				\textbf{Base case:} Already shown for $K = 2$.\\
				
				\textbf{Inductive step:} We will proceed via contrapositive. Suppose the statement is false for $K^*+1$ clusters, where $p_{K^*+1} < 1/2$ by our non-triviality assumption. 
				
				Since $p_{K^*+1} < 1/2$, for $\theta\in\typical{K^*+1}$, we have $\I_{K^*+1} > \J_{K^*+1}$ by definition of the typical set. Equation~\eqref{eq:fbc-kclusters} thus becomes:
				\begin{align}
				FB_C(\theta) =  \left(\frac{\I_1-\J_1}{2}\right)^+ + \ldots + \left(\frac{\I_{K^*+1}-\J_{K^*+1}}{2}\right)
				\end{align}

				If there does not exist $\{y_{\sigma i}\}$ such that \emph{Cluster First Best} is dual feasible, this means that, for all $y_{\sigma i} \geq 0$, there exists a pair $(\theta,\sigma)$, $\theta \in \typical{K^*+1}$, such that
				\begin{align}
				m_{00}(\sigma(\theta))- \sum_i \bar{v}(\sig{i}(\theta))y_{\sigma i} &< \left(\frac{\I_1-\J_1}{2}\right)^+ + \ldots + \left(\frac{\I_{K^*}-\J_{K^*}}{2}\right)^+ + \left(\frac{\I_{K^*+1}-\J_{K^*+1}}{2}\right) \label{eq:k-clusters-contra}
				\end{align}
				
				The goal is to show that the statement must be false for $K^*$ clusters, i.e., for all $y_{\sigma i} \geq 0$, there exists $(\theta,\sigma)$, $\theta \in \typical{K^*}$, such that
				\begin{align*}
				m_{00}(\sigma(\theta))- \sum_i \bar{v}(\sig{i}(\theta))y_{\sigma i} &< \left(\frac{\I_1-\J_1}{2}\right)^+ + \ldots + \left(\frac{\I_{K^*}-\J_{K^*}}{2}\right)^+.
				\end{align*}
				
				We consider two cases:
				\begin{enumerate}
					\item $p_{K^*} < 1/2$: As for Cluster $K^*+1$, for all $\theta \in \typical{K^*+1}$, $\I_{K^*} > \J_{K^*}$ by construction of the typical set.

					In this case, from inequality (\ref{eq:k-clusters-contra}) we obtain:
					\begin{align}
					m_{00}(\sigma(\theta)) - \sum_i \bar{v}(\sig{i}(\theta))y_{\sigma i} &< \left(\frac{\I_1-\J_1}{2}\right)^+ + \ldots + \frac{\I_{K^*}-\J_{K^*}}{2} + \frac{\I_{K^*+1}-\J_{K^*+1}}{2}\nonumber\\
					&= \left(\frac{\I_1-\J_1}{2}\right)^+ + \ldots + \frac{(\I_{K^*}+\I_{K^*+1})-(\J_{K^*}+\J_{K^*+1})}{2}\label{eq:rhs}
					\end{align}

					Consider the right-hand side of inequality (\ref{eq:rhs}). Note that this corresponds exactly to the scenario with $K^*$ clusters with $\tilde{n}_{K^*}$ agents, where $\tilde{n}_{K^*} = n_{K^*}+n_{K^*+1}$. Let $\tilde{\I}_{K^*}, \tilde{\J}_{K^*}$ be the number of type 0 and type 1 agents, respectively, in the $K^*$th cluster. Let $\tilde{\theta}$ be the same realization as $\theta$, with the sole distinction being that the $K^*$th cluster in $\tilde{\theta}$ aggregates the $K^*$-th and $(K^*+1)$-st clusters of realization $\theta$. Thus, $\tilde{\I}_{K^*} = \I_{K^*}+\I_{K^*+1}, \tilde{\J}_{K^*} = \J_{K^*}+\J_{K^*+1}$. Additionally, since $\theta \in \typical{K^*+1}$, we have:
					\begin{align*}
					\I_{K^*} > n_{K^*}/2, \I_{K^*+1} > n_{K^*+1}/2
					\implies \tilde{\I}_{K^*} > \tilde{\J}_{K^*} 
					\implies \tilde{\theta} \in \typical{K^*}
					\end{align*}
					
					Thus, by inequality (\ref{eq:rhs}), there exists $\sigma \in \bSigma $ such that, for any $y_{\sigma i} \geq 0$, for this specific realization $\tilde{\theta}$: 
					\begin{align*}
					m_{00}(\sigma(\theta))- \sum_i \bar{v}(\sig{i}(\theta))y_{\sigma i} < \left(\frac{\I_1-\J_1}{2}\right)^+ + \ldots + \left(\frac{\tilde{\I}_{K^*}-\tilde{\J}_{K^*}}{2}\right)^+,
					\end{align*}
					which implies that $z_{\theta} = \lambda(\theta)\left[\sum_{k=1}^{K^*}\left(\frac{\I_k(\theta)-\J_k(\theta)}{2}\right)^+\right]$ is not feasible in the $K^*$-cluster case. 
					
					\item $p_{K^*} > 1/2$: Since $\theta \in \typical{K^*+1}$, in this case $\I_{K^*} < \J_{K^*}$.
					
					In this case, by inequality (\ref{eq:k-clusters-contra}):
					\begin{align*}
					m_{00}(\sigma(\theta)) - \sum_i \bar{v}(\sig{i}(\theta))y_{\sigma i} &= \left(\frac{\I_1-\J_1}{2}\right)^+ +\left(\frac{\I_{K^*}-\J_{K^*}}{2}\right)^+ + \frac{\I_{K^*+1}-\J_{K^*+1}}{2}\nonumber\\
					&= \left(\frac{(\I_1+\I_{K^*})-(\J_1+\J_{K^*})}{2}\right)^+ + \frac{\I_{K^*+1}-\J_{K^*+1}}{2} 			\end{align*}
					where the last equality follows from the fact that $\I_1 < n_1/2, \I_{K^*} < n_{K^*}/2 \implies \I_1+\I_{K^*} < \J_1+\J_{K^*}$
				\end{enumerate}
				Applying a similar argument as the one in Case 1, we find that, if we are in the $K^*$-cluster case and fix $y_{\sigma i} \geq 0$, for $\theta$ such that $\tilde{\I}_{K^*} = \I_{1}+\I_{K^*},\tilde{\J}_{K^*}=\J_1+\J_K^*$, $z_\theta = \lambda(\theta) \left[\sum_{k=1}^{K^*} \left(\frac{\I_k-\J_k}{2}\right)^+\right]$ is not feasible.
			\end{proof}
		}
		
		This inductive argument gives us many properties about the $K > 2$ setting for free, namely that, for this set of $z_{\theta}, \theta \in \typical{2}$, one can construct a $y_{\sigma i} \in \{0,1\}, \forall \; \sigma \in \bSigma , i \in [n-1]$, such that $(z_{\theta}, y_{\sigma i})$ is dual feasible.  The argument for this is immediate, once one realizes that the $K > 2$ case can be reduced to the $K = 2$ case, by grouping all clusters such that $p_k > 1/2$, and $p_k < 1/2$, respectively.
		
		We derive the upper and lower bounds in an identical fashion to the $K = 2$ case, and obtain our main result:
		\begin{align*}
		\sum_{\theta\in \typical{K}}\lambda(\theta)FB_C(\theta) - o_{n_1+\ldots+n_K}(1) \leq OPT \leq \sum_{\theta\in \typical{K}}\lambda(\theta)FB_C(\theta) + o_{n_1+\ldots+n_K}(1)
		\end{align*} 
	\end{proofof}

	\section{Proof of Lemma~\ref{lem:index-set}}\label{ssex:app-b}

	\begin{proof}
To show the claim, we will do an induction on $0 \leq N \leq N_{\max}$, the number of swaps that have created inter-cluster matches. {At a high level, what we will show is that, for each $\swapsig{N},$ $\theta \in \typical{2}$, one of two events occurs: either the $(N+1)$st swap made the number of `0-0' matches ``large,'' in some sense, relative to $FB_C(\theta)$; or the $(N+1)$st swap decreased the number of `0-0' matches, but the number of `0-0' matches after the $N$th swap was already ``large'' relative to $FB_C(\theta)$.}

	{Before we proceed with the induction, we establish the following facts which we will use throughout. 
	 Recall that we use $\indexswapsig{\swappedind{1}}{N}(\theta),\indexswapsig{\swappedind{2}}{N}(\theta)$ to denote the types of the two agents that were swapped for the $N$th swap, where the subscript beneath the $i$ refers to the cluster in which the agent found herself \emph{after} the swap.

	\begin{prop}\label{prop:same-type-same-num}
		If $\indexswapsig{\swappedind{1}}{N+1}(\theta) = \indexswapsig{\swappedind{2}}{N+1}(\theta)$, then $m_{00}\left(\swapsig{N+1}(\theta)\right) = m_{00}\left(\swapsig{N}(\theta)\right)$.
	\end{prop}
	We omit the proof of this obvious fact. At a high level, it should be clear that swapping agents of the same type has no effect on the objective.
	
	\begin{prop}\label{prop:same-partner-same-num}
		If $\indexswapsig{\swappedind{1}}{N+1}(\theta) \neq \indexswapsig{\swappedind{2}}{N+1}(\theta)$, but $\indexswapsig{m(\swappedind{1})}{N+1}(\theta) = \indexswapsig{m(\swappedind{2})}{N+1}(\theta)$, then $m_{00}\left(\swapsig{N+1}(\theta)\right) \geq m_{00}\left(\swapsig{N}(\theta)\right)$.
	\end{prop}
	
	\begin{proof}
		We assume without loss of generality that $\indexswapsig{\swappedind{1}}{N+1}(\theta) = 1$, $\indexswapsig{\swappedind{2}}{N+1}(\theta) = 0$. This implies that, after the $N$th swap, $\indexswapsig{\swappedind{1}}{N}(\theta) = 0$, $\indexswapsig{\swappedind{2}}{N}(\theta) = 1$.
		We consider two cases:
		\begin{enumerate}
		\item $\indexswapsig{m(\swappedind{1})}{N+1}(\theta) = \indexswapsig{m(\swappedind{2})}{N+1}(\theta) = 0$:  The swap destroyed a `0-0' match in Cluster 1, but created one in Cluster 2.
		\item $\indexswapsig{m(\swappedind{1})}{N+1}(\theta) = \indexswapsig{m(\swappedind{2})}{N+1}(\theta) = 1$: There were no `0-0' matches after the $N$th swap, and none were created in the $(N+1)$-st swap.
		\end{enumerate}
	\end{proof}
	
	\begin{prop}\label{prop:at-most-one-less}
		$$ m_{00}\left(\swapsig{N+1}\right) \geq m_{00}\left(\swapsig{N}\right) - 1\hspace{0.5cm} $$  
	\end{prop}
	
	\begin{proof}
	 Suppose two `0-0' matches were destroyed. Note that a swap affects exactly 4 agents (the two agents swapped and their respective matches). Thus, it had to have been that all 4 swapped agents were type 0 agents. However, if all 4 agents were type 0 agents, then the two swapped agents had the same type, which, by Proposition~\ref{prop:same-type-same-num}, implies that $m_{00}\left(\swapsig{N+1}(\theta)\right) = m_{00}\left(\swapsig{N}(\theta)\right)$. This contradicts the fact that two `0-0' matches were destroyed, and thus completes the proof.
	\end{proof}
	
	\begin{prop}\label{prop:exactly-one-more}
	Under the following scenarios, $\numbadteams{\swapsig{N+1}} = \numbadteams{\swapsig{N}} + 1$:
	\begin{itemize}
		\item $\left(\indexswapsig{\swappedind{1}}{N+1}, \indexswapsig{m(\swappedind{1})}{N+1}\right) = (0,0)$, $\left(\indexswapsig{\swappedind{2}}{N+1}, \indexswapsig{m(\swappedind{2})}{N+1}\right) = (1,1)$
		\item $\left(\indexswapsig{\swappedind{1}}{N+1}, \indexswapsig{m(\swappedind{1})}{N+1}\right) = (1,1)$, $\left(\indexswapsig{\swappedind{2}}{N+1}, \indexswapsig{m(\swappedind{2})}{N+1}\right) = (0,0)$
	\end{itemize}
	\end{prop}
	
	\begin{proof}
		We only show the claim for the first scenario, as the proof for the second is identical. 
		\begin{align*}
		\left(\indexswapsig{\swappedind{1}}{N+1}, \indexswapsig{m(\swappedind{1})}{N+1}\right) = (0,0), \left(\indexswapsig{\swappedind{2}}{N+1}, \indexswapsig{m(\swappedind{2})}{N+1}\right) = (1,1) \implies \left(\indexswapsig{\swappedind{1}}{N}, \indexswapsig{m(\swappedind{1})}{N}\right) = (1,0), \left(\indexswapsig{\swappedind{2}}{N}, \indexswapsig{m(\swappedind{2})}{N}\right) = (0,1).
		\end{align*}
		Thus, from the $N$th swap to the $(N+1)$st swap, one additional `0-0' match was created.
	\end{proof}

	}
	
	{For convenience, we {henceforth abuse notation and suppress the dependence on $\theta$. It should be clear from context that the variables to which we refer depend on the realization $\theta$. In particular, $\sigma_i$ will no longer denote the identity of the $i$th agent under ordering $\sigma$, but the \emph{type} of the $i$th agent in ordering $\sigma$. Similarly, $m_{00}(\sigma)$ denotes the number of `0-0' matches induced by ordering $\sigma$ \emph{under realization $\theta$}.}  \\	

\noindent\textbf{Base cases (even and odd):}
		\begin{enumerate}
			\item $N=0$:
			
			$N=0$ corresponds to the identity permutation, $\sigma^0$. Consider $\theta^* \in \arg\min_{\theta} m_{00}\left(\sigma{\thetablank}\right)$. This realization corresponds to the maximal pairing of 1s and 0s, in which we case we obtain $m_{00}(\sigma^0(\theta^*)) = \frac{\I_2-\J_2}{2}$. Thus, when $N=0$, for all $\theta \in \typical{2}$:
			\begin{align*}
			m_{00}(\sigma^0{\thetablank})-\sum_{i\in\mathcal{I}_{\sigma^0}}\bar{v}(\sigma_0^i{\thetablank})&=m_{00}(\sigma^0{\thetablank})\geq\frac{\I_2-\J_2}{2}
			\end{align*}
			\item $N=1$:
			
			For $N=1$, one agent from Cluster 1 was placed in Cluster 2, and one agent from Cluster 2 was placed in Cluster 1. By construction, $\indexset{\sigma}{1}$ is such that
			\begin{align*}
			\sum_{i\in\indexset{\sigma}{1}} \bar{v}(\indexswapsig{i}{1}{\thetablank}) &= \sum_{i\in\mathcal{I}_{\sigma^0}} \bar{v}(\indexswapsig{i}{0}{\thetablank})+\indexswapsig{\swappedind{1}}{1}{\thetablank}-\indexswapsig{\swappedind{2}}{1}{\thetablank}= \indexswapsig{\swappedind{1}}{1}{\thetablank}-\indexswapsig{\swappedind{2}}{1}{\thetablank}
			\end{align*}

			Thus, for all $\theta \in \typical{2}$, we have:
			\begin{align}\label{n=1}
			&m_{00}(\swapsig{1}{\thetablank}) - \sum_{i\in\indexset{\sigma}{1}}\bar{v}(\indexswapsig{i}{1}{\thetablank}) 
			=m_{00}(\swapsig{1}{\thetablank}) - \indexswapsig{\swappedind{1}}{1}{\thetablank} + \indexswapsig{\swappedind{2}}{1}{\thetablank}. 
			\end{align}
			For $\theta$ such that $ \indexswapsig{\swappedind{1}}{1}{(\theta)} = \indexswapsig{\swappedind{2}}{1}{(\theta)}$, by Proposition~\ref{prop:same-type-same-num}, $m_{00}(\swapsig{1}{\thetablank}) = m_{00}(\sigma^0{\thetablank})$. We just showed above that for $N = 0$ $ m_{00}(\sigma^0{\thetablank}) \geq \frac{\I_2-\J_2}{2}$.
						
			We now consider $\theta$ such that $\indexswapsig{\swappedind{1}}{1}{(\theta)} \neq \indexswapsig{\swappedind{2}}{1}{(\theta)}$. There are two possibilities:
			\begin{enumerate}\setlength\itemsep{1em}
				\item\label{case:base-1a} {} $\indexswapsig{\swappedind{1}}{1}{\thetablank} = 0, \indexswapsig{\swappedind{2}}{1}{\thetablank} = 1$:
				
				Plugging this into~\eqref{n=1}:
				\begin{align*}
				m_{00}(\swapsig{1}{\thetablank}) - \sum_{i\in\indexset{\sigma}{1}}\bar{v}(\indexswapsig{i}{1}{\thetablank})  = m_{00}(\swapsig{1}{\thetablank})+1
				\end{align*}
				
				By Proposition~\ref{prop:at-most-one-less}, $m_{00}(\swapsig{1}{\thetablank}) \geq m_{00}(\swapsig{0}{\thetablank})-1$. Using this:
				\begin{align*}
				m_{00}(\swapsig{1}{\thetablank}) &\geq m_{00}(\sigma^0{\thetablank})-1\\
				\implies m_{00}(\swapsig{1}{\thetablank})+1 &\geq m_{00}(\sigma^0{\thetablank})
				\geq \frac{\I_2-\J_2}{2}
				\end{align*}
				where the last inequality follows from $N=0$.

				\item\label{case:base-1b} {} $\indexswapsig{\swappedind{1}}{1}{\thetablank} = 1, \indexswapsig{\swappedind{2}}{1}{\thetablank} = 0 $:
				
				Plugging this into~\eqref{n=1}:
				\begin{align}
			m_{00}(\swapsig{1}{\thetablank}) - \sum_{i\in\indexset{\sigma}{1}}\bar{v}(\indexswapsig{i}{1}{\thetablank})  = m_{00}(\swapsig{1}{\thetablank})-1 \label{caseb}
				\end{align}

				We consider four cases, representing the possible types of the partners of the swapped agents. 	 		 
				\begin{enumerate}
					\item $\indexswapsig{m(\swappedind{1})}{1}{\thetablank}=\indexswapsig{m(\swappedind{2})}{1}{\thetablank}=1$:
					
					By Proposition~\ref{prop:same-partner-same-num}, $m_{00}(\swapsig{1}{\thetablank}) =m_{00}(\sigma^0{\thetablank})$. We also have the following fact about the number of `0-0' matches under the identity permutation.
					
					\begin{prop}\label{case:base-2-b-i}
					$$ m_{00}(\sigma^0{\thetablank}) \geq \frac{\I_2-\J_2}{2}+1$$
					\end{prop}
					
					\begin{proof}
				\begin{align*}
					\left(\indexswapsig{\swappedind{1}}{1}{\thetablank}, \indexswapsig{m(\swappedind{1})}{1}{\thetablank}\right) = (1,1), \left(\indexswapsig{\swappedind{2}}{1}{\thetablank}, \indexswapsig{m(\swappedind{2})}{1}{\thetablank}\right) = (0,1) \implies \left(\indexswapsig{\swappedind{1}}{0}{\thetablank}, \indexswapsig{m(\swappedind{1})}{0}{\thetablank}\right) = (0,1), \left(\indexswapsig{\swappedind{2}}{0}{\thetablank}, \indexswapsig{m(\swappedind{2})}{0}{\thetablank}\right) = (1,1)
					\end{align*}
					This implies that, under the identity permutation, two type 1 agents were matched together. Recall that the minimum number of `0-0' matches is achieves by the maximum number of consecutive `1-0's under any permutation. Given that we know for certain that two 1s are matched together under the identity permutation (and hence consecutive), only $\J_2-2$ type 1 agents remain to match with type 0 agents. Thus, the \emph{minimum} number of `0-0' matches is $\frac{\I_2-\left(\J_2-2\right)}{2} = \frac{\I_2-\J_2}{2}+1$.
					\end{proof}
						
Using these two observations and plugging into Equation~\eqref{caseb}, we obtain:
					\begin{align*}
						m_{00}(\swapsig{1}{\thetablank}) - \sum_{i\in\indexset{\sigma}{1}}\bar{v}(\indexswapsig{i}{1}{\thetablank}) &= m_{00}(\sigma^0{\thetablank})-1
					\geq \left(\frac{\I_2-\J_2}{2}+1\right)-1
					= \frac{\I_2-\J_2}{2}.
					\end{align*}
					
					\item $\indexswapsig{m(\swappedind{1})}{1}{\thetablank}=\indexswapsig{m(\swappedind{2})}{1}{\thetablank}=0$:

					By Proposition~\ref{prop:same-partner-same-num}, $m_{00}(\swapsig{1}{\thetablank}) =m_{00}(\sigma^0{\thetablank})$. Moreover, we have the following fact regarding the number of `0-0' matches under the identity permutation. 
										
					\begin{prop}\label{case:base-2-b-ii}
						$$\numbadteams{\swapsig{0}} \geq \frac{\I_2-\J_2}{2} + 1.$$
					\end{prop}
					
					\begin{proof} 
					\begin{align*}
					\left(\indexswapsig{\swappedind{1}}{1}{\thetablank}, \indexswapsig{m(\swappedind{1})}{1}{\thetablank}\right) = (1,0), \left(\indexswapsig{\swappedind{2}}{1}{\thetablank}, \indexswapsig{m(\swappedind{2})}{1}{\thetablank}\right) = (0,0) \implies \left(\indexswapsig{\swappedind{1}}{0}{\thetablank}, \indexswapsig{m(\swappedind{1})}{0}{\thetablank}\right) = (0,0), \left(\indexswapsig{\swappedind{2}}{0}{\thetablank}, \indexswapsig{m(\swappedind{2})}{0}{\thetablank}\right) = (1,0)
					\end{align*}
					Thus, in the identity permutation there was \emph{at least} one `0-0' match in Cluster 1. The number of `0-0' matches in Cluster 2 is still lower bounded by $\frac{\I_2-\J_2}{2}$, the optimal matching of type 1 and type 0 agents in this cluster. Putting these two together we obtain the lower bound on the number of `0-0' matches in the identity permutation.
					\end{proof}
					
					Putting the above two facts together and plugging into Equation~\eqref{caseb}:
					\begin{align*}
					m_{00}(\swapsig{1}{\thetablank}) - \sum_{i\in\indexset{\sigma}{1}}\bar{v}(\indexswapsig{i}{1}{\thetablank})&= m_{00}(\sigma^0{\thetablank})-1
					\geq \left(\frac{\I_2-\J_2}{2}+1\right)-1
					= \frac{\I_2-\J_2}{2}
					\end{align*}

					\item $\indexswapsig{m(\swappedind{1})}{1}{\thetablank}=1, \indexswapsig{m(\swappedind{2})}{1}{\thetablank}=0$:
					
					By Proposition~\ref{prop:exactly-one-more}, $m_{00}(\swapsig{1}{\thetablank}) = m_{00}(\sigma^0{\thetablank})+1$. Using this fact, Equation~\eqref{caseb} becomes:
					\begin{align*}
					m_{00}(\swapsig{1}{\thetablank}) - \sum_{i\in\indexset{\sigma}{1}}\bar{v}(\indexswapsig{i}{1}{\thetablank}) &= \left(m_{00}(\sigma^0{\thetablank})+1\right)-1
					=m_{00}(\sigma^0{\thetablank})
					\geq \frac{\I_2-\J_2}{2}.
					\end{align*}
					
					\item $\indexswapsig{m(\swappedind{1})}{1}{\thetablank}=0, \indexswapsig{m(\swappedind{2})}{1}{\thetablank}=1$:
					
					By Proposition~\ref{prop:at-most-one-less},  $m_{00}(\swapsig{1}{\thetablank}) \geq m_{00}(\sigma^0{\thetablank})-1$. We also have the following fact about the number of `0-0' matches under the identity permutation.
					
					\begin{prop}
					$$ \numbadteams{\swapsig{0}} \geq \frac{\I_2-\J_2}{2}+2$$
					\end{prop}
					
					\begin{proof}
					\begin{align*}
	\left(\indexswapsig{\swappedind{1}}{1}{\thetablank}, \indexswapsig{m(\swappedind{1})}{1}{\thetablank}\right) = (1,0), \left(\indexswapsig{\swappedind{2}}{1}{\thetablank}, \indexswapsig{m(\swappedind{2})}{1}{\thetablank}\right) = (0,1) \implies \left(\indexswapsig{\swappedind{1}}{0}{\thetablank}, \indexswapsig{m(\swappedind{1})}{0}{\thetablank}\right) = (0,0), \left(\indexswapsig{\swappedind{2}}{0}{\thetablank}, \indexswapsig{m(\swappedind{2})}{0}{\thetablank}\right) = (1,1)
					\end{align*}
					
					Combining the arguments used to prove Propositions~\ref{case:base-2-b-i} and~\ref{case:base-2-b-ii}, we obtain the lower bound $\numbadteams{\swapsig{0}} \geq 1 + \frac{\I_2-(\J_2-2)}{2}$, which gives us the result.
					\end{proof}
															
					Putting these facts together, Equation~\eqref{caseb} becomes:
					\begin{align*}
					m_{00}(\swapsig{1}{\thetablank})-1 &\geq \left(m_{00}(\sigma^0{\thetablank})-1\right)-1
					\geq \left(\frac{\I_2-\J_2}{2}+2\right)-2
					= \frac{\I_2-\J_2}{2}
					\end{align*}
				\end{enumerate}
			\end{enumerate}
		\end{enumerate}

\noindent\textbf{Inductive step (even):} Let $1 \leq N^* \leq N_{\max}$ be an arbitrary odd integer. We will show that, if $\indexset{\sigma}{N^*}$ satisfies Inequality~\eqref{eq:induc-ineq}, then so does $\indexset{\sigma}{N^*+1}$.
		
		Recall that, by construction, $\indexset{\sigma}{N^*+1}$ satisfies
		\begin{align}\label{eq:index-n+1}
		\sum_{i\in\indexset{\sigma}{N^*+1}}\bar{v}(\indexswapsig{i}{N^*+1}{\thetablank}) = \sum_{i\in\indexset{\sigma}{N^*}}\bar{v}(\indexswapsig{i}{N^*}{\thetablank})-\indexswapsig{m(\swappedind{1})}{N^*+1}{\thetablank}+\indexswapsig{m(\swappedind{2})}{N^*+1}{\thetablank}.
		\end{align}

		As before, we consider four scenarios, depending on the types of the swapped agents' matches.
		\begin{enumerate}\setlength\itemsep{1em}
			\item $\indexswapsig{m(\swappedind{1})}{N^*+1}{\thetablank} = 1, \indexswapsig{m(\swappedind{2})}{N^*+1}{\thetablank}=1$:
			
			By Equation~\eqref{eq:index-n+1}:
			\begin{align}\label{eq:induc-case1}
			m_{00}(\swapsig{N^*+1}{\thetablank}) -\sum_{i\in\indexset{\sigma}{N^*+1}}\bar{v}(\indexswapsig{i}{N^*+1}{\thetablank})
			=m_{00}(\swapsig{N^*+1}{\thetablank}) -\sum_{i\in\indexset{\sigma}{N^*}}\bar{v}(\indexswapsig{i}{N^*}{\thetablank})
			\end{align}
			
			By Proposition~\ref{prop:same-partner-same-num}, $m_{00}(\swapsig{N^*+1}{\thetablank}) =m_{00}(\swapsig{N^*}{\thetablank})$. Plugging this into Equation~\eqref{eq:induc-case1}:
			\begin{align*}
			m_{00}(\swapsig{N^*+1}{\thetablank}) -\sum_{i\in\indexset{\sigma}{N^*+1}}\bar{v}(\indexswapsig{i}{N^*+1}{\thetablank})
			&= m_{00}(\swapsig{N^*}{\thetablank}) -\sum_{i\in\indexset{\sigma}{N^*}}\bar{v}(\indexswapsig{i}{N^*}{\thetablank})
			\geq \frac{\I_2-\J_2}{2}
			\end{align*}
			by the inductive hypothesis.
			\item $\indexswapsig{m(\swappedind{1})}{N^*+1}{\thetablank} = 0, \indexswapsig{m(\swappedind{2})}{N^*+1}{\thetablank}=0$:
			
			By Proposition~\ref{prop:same-partner-same-num}, $m_{00}(\swapsig{N^*+1}{\thetablank}) = m_{00}(\swapsig{N^*}{\thetablank})$. Using this fact, we have:
			\begin{align*}
		m_{00}(\swapsig{N^*+1}{\thetablank}) -\sum_{i\in\indexset{\sigma}{N^*+1}}\bar{v}(\indexswapsig{i}{N^*+1}{\thetablank})
			&=m_{00}(\swapsig{N^*+1}{\thetablank}) -\sum_{i\in\indexset{\sigma}{N^*}}\bar{v}(\indexswapsig{i}{N^*}{\thetablank})\\
			&=m_{00}(\swapsig{N^*}{\thetablank})-\sum_{i\in\indexset{\sigma}{N^*}}\bar{v}(\indexswapsig{i}{N^*}{\thetablank})
			\geq \frac{\I_2-\J_2}{2},
			\end{align*}
			where the inequality follows from the inductive hypothesis. \\
			
			\item $\indexswapsig{m(\swappedind{1})}{N^*+1}{\thetablank}=1, \indexswapsig{m(\swappedind{2})}{N^*+1}{\thetablank}=0$:
			\begin{align}\label{eq:induc-case3}
			m_{00}(\swapsig{N^*+1}{\thetablank}) -\sum_{i\in\indexset{\sigma}{N^*+1}}\bar{v}(\indexswapsig{i}{N^*+1}{\thetablank})
			= m_{00}(\swapsig{N^*+1}{\thetablank}) -\sum_{i\in\indexset{\sigma}{N^*}}\bar{v}(\indexswapsig{i}{N^*}{\thetablank}) +1
			\end{align}
			
			By Proposition \ref{prop:at-most-one-less},  $m_{00}(\swapsig{N^*+1}){\thetablank} \geq m_{00}(\swapsig{N^*}{\thetablank})-1$. Using this fact and plugging into Equation~\eqref{eq:induc-case3}, we obtain:
			
			\begin{align*}
		m_{00}(\swapsig{N^*+1}{\thetablank}) -\sum_{i\in\indexset{\sigma}{N^*+1}}\bar{v}(\indexswapsig{i}{N^*+1}{\thetablank})
			&\geq \left(m_{00}(\swapsig{N^*}{\thetablank})-1\right) -\sum_{i\in\indexset{\sigma}{N^*}}\bar{v}(\indexswapsig{i}{N^*}{\thetablank}) +1\\
			&= m_{00}(\swapsig{N^*}{\thetablank}) -\sum_{i\in\indexset{\sigma}{N^*}}\bar{v}(\indexswapsig{i}{N^*}{\thetablank})\\
			&\geq \frac{\I_2-\J_2}{2}
			\end{align*}
			by the inductive hypothesis.
			
			\item $\indexswapsig{m(\swappedind{1})}{N^*+1}{\thetablank}=0,\indexswapsig{m(\swappedind{2})}{N^*+1}{\thetablank}=1$:
			
			In this case,
			\begin{align*}
			m_{00}(\swapsig{N^*+1}{\thetablank}) -\sum_{i\in\indexset{\sigma}{N^*+1}}\bar{v}(\indexswapsig{i}{N^*+1}{\thetablank})
			= m_{00}(\swapsig{N^*+1}{\thetablank}) -\sum_{i\in\indexset{\sigma}{N^*}}\bar{v}(\indexswapsig{i}{N^*}{\thetablank})-1
			\end{align*}
			We consider four sub-cases, depending on the types of the swapped individuals in the $(N^*+1)$st swap:
			\begin{enumerate}[i.]
				\item  $\indexswapsig{\swappedind{1}}{N^*+1}{\thetablank} = 1, \indexswapsig{\swappedind{2}}{N^*+1}{\thetablank}= 1 $:
				
				By Proposition~\ref{prop:same-type-same-num}, $m_{00}(\swapsig{N^*+1}{\thetablank}) = m_{00}(\swapsig{N^*}{\thetablank}).$ We will show the following additional fact.
				
				\begin{prop}\label{case:induc-4-i}
				$$m_{00}(\swapsig{N^*}{\thetablank})-\sum_{i\in\indexset{\sigma}{N^*}}\bar{v}(\indexswapsig{i}{N^*}{\thetablank}) \geq \frac{\I_2-\J_2}{2} + 1$$
				\end{prop}
				
				\begin{proof}
				
				\begin{align*}
					\left(\indexswapsig{\swappedind{1}}{N^*+1}{\thetablank}, \indexswapsig{m(\swappedind{1})}{N^*+1}{\thetablank}\right) = (1,0), \left(\indexswapsig{\swappedind{2}}{N^*+1}{\thetablank}, \indexswapsig{m(\swappedind{2})}{N^*+1}{\thetablank}\right) = (1,1) \implies \left(\indexswapsig{\swappedind{1}}{N^*}{\thetablank}, \indexswapsig{m(\swappedind{1})}{N^*}{\thetablank}\right) = (1,0), \left(\indexswapsig{\swappedind{2}}{N^*}{\thetablank}, \indexswapsig{m(\swappedind{2})}{N^*}{\thetablank}\right) = (1,1)
					\end{align*}
					
We earlier argued that the agents swapped in the $(N^*+1)$st swap had never been swapped before (and similarly for their partners). Thus, the $N^*$th swap can equivalently be viewed as an operation on $n_1$ agents in Cluster 1, and $n_2-2$ agents in Cluster 2, with $\J_2-2$ of them being type 1 agents (given that we know that two of the $\J$ type 1 agents were the swapped agent and her partner). Thus, by the inductive hypothesis
\begin{align*}
m_{00}(\swapsig{N^*}{\thetablank})-\sum_{i\in\indexset{\sigma}{N^*}}\bar{v}(\indexswapsig{i}{N^*}{\thetablank}) \geq \frac{\I_2-\left(\J_2-2\right)}{2} = \frac{\I_2-\J_2}{2} + 1.
\end{align*}				
				\end{proof}
				
				Putting these two observations together, we have:
				\begin{align*}
					m_{00}(\swapsig{N^*+1}{\thetablank}) -\sum_{i\in\indexset{\sigma}{N^*+1}}\bar{v}(\indexswapsig{i}{N^*+1}{\thetablank})
				\geq \left(\frac{\I_2-\J_2}{2} + 1\right)-1 = \frac{\I_2-\J_2}{2}
				\end{align*}

				\item {} $\indexswapsig{\swappedind{1}}{N^*+1}{\thetablank} = 0, \indexswapsig{\swappedind{2}}{N^*+1}{\thetablank}= 0 $: 
					
				By Proposition~\ref{prop:same-type-same-num}, $m_{00}(\swapsig{N^*+1}{\thetablank}) = m_{00}(\swapsig{N^*}{\thetablank}).$ We will show the following additional fact.
				
				\begin{prop}\label{case:induc-4-ii}
				$$ m_{00}(\swapsig{N^*}{\thetablank}) -\sum_{i\in\indexset{\sigma}{N^*}}\bar{v}(\indexswapsig{i}{N^*}{\thetablank}) \geq \frac{\I_2-\J_2}{2} + 1$$
				\end{prop}
				
				\begin{proof}
				\begin{align*}
					\left(\indexswapsig{\swappedind{1}}{N^*+1}{\thetablank}, \indexswapsig{m(\swappedind{1})}{N^*+1}{\thetablank}\right) = (0,0), \left(\indexswapsig{\swappedind{2}}{N^*+1}{\thetablank}, \indexswapsig{m(\swappedind{2})}{N^*+1}{\thetablank}\right) = (0,1) \implies \left(\indexswapsig{\swappedind{1}}{N^*}{\thetablank}, \indexswapsig{m(\swappedind{1})}{N^*}{\thetablank}\right) = (0,0), \left(\indexswapsig{\swappedind{2}}{N^*}{\thetablank}, \indexswapsig{m(\swappedind{2})}{N^*}{\thetablank}\right) = (0,1)
				\end{align*}
				
This implies that, after the $N^*$th swap, there was a `0-0' match in Cluster 1. Using the same argument as in the previous case, we can view the permutation $\sigma_{N^*}$ as one on $n_2$ agents in Cluster 2 (with $\I_2, \J_2$ type 0 and type 1 agents, respectively), and $n_1-2$ agents in Cluster 1, with $\I_1-2$ type 0 agents (since we know that the $(N^*+1)$st pair to be swapped was never considered for $N^*$), plus an extra `0-0' match in Cluster 1. Thus, by the inductive hypothesis
\begin{align*}
m_{00}(\swapsig{N^*}{\thetablank})-\sum_{i\in\indexset{\sigma}{N^*}}\bar{v}(\indexswapsig{i}{N^*}{\thetablank}) \geq 1 + \frac{\I_2-\J_2}{2}.
\end{align*}				
\end{proof}

Putting these two facts together, we have:
																			
				\begin{align*}
				m_{00}(\swapsig{N^*+1}{\thetablank}) -\sum_{i\in\indexset{\sigma}{N^*+1}}\bar{v}(\indexswapsig{i}{N^*+1}{\thetablank})
				&= m_{00}(\swapsig{N^*}{\thetablank}) -\sum_{i\in\indexset{\sigma}{N^*}}\bar{v}(\indexswapsig{i}{N^*}{\thetablank})-1\\
				&\geq \left(\frac{\I_2-\J_2}{2}+1\right)-1
				= \frac{\I_2-\J_2}{2}
				\end{align*} 
				
				\item  {} $\indexswapsig{\swappedind{1}}{N^*+1}{\thetablank} = 0, \indexswapsig{\swappedind{2}}{N^*+1}{\thetablank}= 1$:
				
				By Proposition~\ref{prop:exactly-one-more}, $m_{00}(\swapsig{N^*+1}{\thetablank})= m_{00}(\swapsig{N^*}{\thetablank})+1$. Using this fact:
				\begin{align*}
					m_{00}(\swapsig{N^*+1}{\thetablank}) -\sum_{i\in\indexset{\sigma}{N^*+1}}\bar{v}(\indexswapsig{i}{N^*+1}{\thetablank})
				&= \left(m_{00}(\swapsig{N^*}{\thetablank})+1\right) -\sum_{i\in\indexset{\sigma}{N^*}}\bar{v}(\indexswapsig{i}{N^*}{\thetablank})-1\\
				&= m_{00}(\swapsig{N^*}{\thetablank})-\sum_{i\in\indexset{\sigma}{N^*}}\bar{v}(\indexswapsig{i}{N^*}{\thetablank})\\
				&\geq \frac{\I_2-\J_2}{2}
				\end{align*}
				by the inductive hypothesis.

				\item {} $\indexswapsig{\swappedind{1}}{N^*+1}{\thetablank} = 1, \indexswapsig{\swappedind{2}}{N^*+1}{\thetablank}= 0 $ :
				
				By Proposition~\ref{prop:at-most-one-less}, $m_{00}(\swapsig{N^*+1}{\thetablank}) \geq m_{00}(\swapsig{N^*}{\thetablank})-1$. We show the following additional fact.
				
				\begin{prop}
				$$m_{00}\left(\swapsig{N^*}{\thetablank}\right)-\sum_{i\in\indexset{\sigma}{N^*}}\bar{v}(\indexswapsig{i}{N^*}{\thetablank}) \geq \frac{\I_2-\J_2}{2}$$
				\end{prop}
				
				\begin{proof}
			\begin{align*}
					\left(\indexswapsig{\swappedind{1}}{N^*+1}{\thetablank}, \indexswapsig{m(\swappedind{1})}{N^*+1}{\thetablank}\right) = (1,0), \left(\indexswapsig{\swappedind{2}}{N^*+1}{\thetablank}, \indexswapsig{m(\swappedind{2})}{N^*+1}{\thetablank}\right) = (0,1) \implies \left(\indexswapsig{\swappedind{1}}{N^*}{\thetablank}, \indexswapsig{m(\swappedind{1})}{N^*}{\thetablank}\right) = (0,0), \left(\indexswapsig{\swappedind{2}}{N^*}{\thetablank}, \indexswapsig{m(\swappedind{2})}{N^*}{\thetablank}\right) = (1,1)
				\end{align*}
				
				Once we have established the above, the proof of the claim follows directly from a combination of the arguments use to prove Propositions~\ref{case:induc-4-i} and~\ref{case:induc-4-ii}. 
				\end{proof}
								
Putting this all together, we obtain:				
				\begin{align*}
					m_{00}\left(\swapsig{N^*+1}{\thetablank}\right) -\sum_{i\in\indexset{\sigma}{N^*+1}}\bar{v}(\indexswapsig{i}{N^*+1}{\thetablank})
				&\geq\left(m_{00}\left(\swapsig{N^*}{\thetablank}\right)-1\right)-\sum_{i\in\indexset{\sigma}{N^*}}\bar{v}(\indexswapsig{i}{N^*}{\thetablank})-1\\
				&= m_{00}\left(\swapsig{N^*}{\thetablank}\right)-\sum_{i\in\indexset{\sigma}{N^*}}\bar{v}(\indexswapsig{i}{N^*}{\thetablank})-2\\
				&\geq \left(\frac{\I_2-\J_2}{2}+2\right)-2 = \frac{\I_2-\J_2}{2}.
				\end{align*}
			\end{enumerate}
		\end{enumerate}
		
		\textbf{Inductive step (odd):} 
		We omit the proof of this inductive step, since it is  analogous to the case where $N=1$, and the case-by-case analysis is identical to what was done for the even inductive step.
	} \\
		
		Thus, we have shown that this choice of $\indexset{\sigma}{N}$ satisfies Equation~(\ref{eq:to-show}), and Lemma~\ref{lem:fixed-num} immediately follows. \end{proof}
\section{Auxiliary Proofs}\label{sec:aux-proofs}

\begin{proofof}{Proposition~\ref{prop:util-agent}}
	Define $\EE_{\sigma}[\cdot]=\EE\left[\cdot\mid\mF^i_{\sigma}\right]$, and $\PP_\sigma\left[\cdot\right] = \PP\left[\cdot\mid\mF^i_\sigma\right] $. Now consider $\EE_{\sigma}\left[u(\theta_i,\theta_j)\right]$, for arbitrary $j\in[n]$. We have:
	\begin{align}
	\EE_{\sigma}\left[u(\theta_i,\theta_j)\right] 
	=
	& u(0,0)\PP_{\sigma}\left[\theta_i=0,\theta_{j} = 0\right] + u(1,1)\PP_{\sigma}\left[\theta_i=1,\theta_{j} = 1\right] \notag \\
	& + u(1,0)\Big(\PP_{\sigma}\left[\theta_i=1,\theta_{j} = 0\right]\notag+\PP_{\sigma}\left[\theta_i=0,\theta_{j} =1\right]\Big)\notag\\
	=
	&u(0,0)\PP_{\sigma}\left[\theta_i=0\mid\theta_j=0\right]\PP_{\sigma}\left[\theta_j=0\right] + u(1,1)\PP_{\sigma}\left[\theta_i=1\mid\theta_j=1\right]\PP_{\sigma}\left[\theta_j=1\right]\notag\\
	&+u(1,0)\Big(\PP_{\sigma}\left[\theta_i=0\mid\theta_j=1\right]\PP_{\sigma}\left[\theta_j=1\right]+\PP_{\sigma}\left[\theta_i=1\mid\theta_j=0\right]\PP_{\sigma}\left[\theta_j=0\right]\Big)\notag
	\end{align}
	{For ease of notation, we} define $q = \PP_{\sigma}\left[\theta_j=1\right] = \EE_{\sigma}\left[\theta_j\right]$, and $(p_0,p_1)$ such that $p_k = \PP_{\sigma}\left[{\theta_i=1}\mid{\theta_j=k}\right]$. Additionally, let $g(q,p_0,p_1) = \EE_{\sigma}\left[u(\theta_i,\theta_j)\right]$. Substituting above, we get:
	\begin{align*}
	g(q,p_0,p_1) = u(1,1)p_1q+u(1,0)\left((1-p_1)q+p_0(1-q)\right)+u(0,0)(1-p_0)(1-q)
	\end{align*}
	
	Our goal is to show that $g(q,p_0,p_1)$ is monotone increasing in $q$. For this, we take the partial derivative with respect to $q$:
	\begin{align*}
	\frac{\partial g}{\partial q} &= u(1,1)p_1+u(1,0)(1-p_0-p_1)-u(0,0)(1-p_0)\\
	&=p_1\left(u(1,1)-u(1,0)\right)+(1-p_0)(u(1,0)-u(0,0))
	\geq 0,
	\end{align*}
	where the inequality follows from the fact that $u(1,1) > u(1,0) > u(0,0)$. 
\end{proofof}\\

\begin{proofof}{Proposition~\ref{prop:matching}}
	We first show the $(\implies)$ direction. Suppose constraint~\eqref{eq:correct-ordering} is satisfied. For agent $i= \sig{1}$, $\arg\max_j\EEc{\theta_j}{\mathcal{F}^i_{\sigma}} = \sig{2}$. Similarly, for agent $i' = \sig{2}$, \; $\arg\max_j\EEc{\theta_j}{\mathcal{F}^{i'}_{\sigma}} = \sig{1}$.  Thus, agents $\sig{1}$ and $\sig{2}$ will match in $m(\sigma)$.  Replicating this logic, agents will never want to match with agents who are ranked below their match which ordering $\sigma$ induced. Though constraint~(\ref{eq:correct-ordering}) leaves open the possibility that an agent may want to match with $\sig{j^*}$ such that $\sig{j^*} < \sig{m(i)}$ (i.e. ranked above their match in the above ordering), such a match would never occur since $\sig{j^*}$ has no incentive to deviate and match with someone ranked lower than her own partner. Thus, the resulting matching is stable.
	
	$(\impliedby)$: Suppose there exists an ordering $\sigma$ and two agents $\sig{i}, \sig{j}$ such that $\sig{i} < \sig{j}$ but $\EEc{\sigma_i(\theta)}{\mathcal{F}^{m(i)}_{\sigma}} < \EEc{\sigma_j(\theta)}{\mathcal{F}^{m(i)}_{\sigma}}$. Since $\EEc{\sigma_i(\theta)}{\mathcal{F}^{m(i)}_{\sigma}} < \EEc{\sigma_j(\theta)}{\mathcal{F}^{m(i)}_{\sigma}}$, agent $\sig{m(i)}$ would rather match with agent $\sig{j}$. Additionally, since $\EEc{\sig{m(i)}(\theta)}{\mathcal{F}^j_{\sigma}} \geq \EEc{\sig{m(j)}(\theta)}{\mathcal{F}^j_{\sigma}}$, agent $\sig{j}$ would rather match with agent $\sig{m(i)}$ than agent $\sig{m(j)}$. Thus, $(\sig{m(i)},\sig{j})$ constitute a blocking pair, and $m(\sigma)$ is unstable.
\end{proofof} \\

\begin{proofof}{Lemma~\ref{lem:convexity}}
	Consider any realization $\theta$ and announced ordering $\sigma$, and let $m_{11}, m_{10}, m_{00}$ denote the number `1-1', `1-0', and `0-0' matches induced by $\sigma$ for this realization. Let $\mathcal{W}$ denote the expected social welfare for any signaling policy.
	\begin{align}\label{eq:one-degree}
	\mathcal{W} &= \EE\left[\sum_{i=1}^n u\left(\sig{i}(\theta), \sig{m(i)}(\theta)\right)\right] \notag \\
	&=u(1,1)\,\EE\left[\sum_{i=1}^n \mathds{1}\left\{\sig{i}(\theta) = 1,\sig{m(i)}(\theta) = 1\right\}\right]+u(0,0)\,\EE\left[\sum_{i=1}^n \mathds{1}\left\{\sig{i}(\theta) = 0,\sig{m(i)}(\theta) = 0\right\}\right] \notag \\
	&\hspace{1cm} + u(1,0) \,\EE\left[\sum_{i=1}^n \left(\mathds{1}\left\{\sig{i}(\theta) = 1,\sig{m(i)}(\theta) = 0\right\}+\mathds{1}\left\{\sig{i}(\theta) = 0,\sig{m(i)}(\theta) = 1\right\}\right)\right] \notag \\
	&= 2\bigg(u(1,1)\,\expectation{m_{11}} + u(0,0)\,\expectation{m_{00}} + u(1,0)\,\expectation{m_{10}}\bigg)
	\end{align}
	
	Recall that $h(\theta), \ell(\theta)$ denote the number of type 1 and type 0 agents, respectively. We can re-write $m_{10}$ and $m_{00}$ as a function of $m_{11}$:
	\begin{equation*}
	\begin{aligned}	
	m_{10} &= h(\theta)-2m_{11}  \\
	m_{00} &= \frac{\ell(\theta) - m_{10}}{2} = \frac{\ell(\theta)-h(\theta)}{2} + m_{11}
	\end{aligned}
	\end{equation*}
	
	Plugging this into Equation~\eqref{eq:one-degree}:
	\begin{align*}
	\mathcal{W} = 2 \,\expectation{m_{11}}\bigg(u(1,1) +u(0,0) - 2u(1,0)\bigg) + \text{cst}
	\end{align*}
	
	For convex $u$, $u(1,1) + u(0,0) - 2u(1,0) \geq 0$, implying $\mathcal{W}$ is \emph{increasing} in $\mathbb{E}[m_{11}]$ (and consequently increasing in $\expectation{m_{00}}$, decreasing in $\expectation{m_{10}}$). For strictly concave $u$, $u(1,1) + u(0,0)- 2u(1,0) < 0$, which implies $\mathcal{W}$ is \emph{decreasing} in $\expectation{m_{11}}$ (and consequently decreasing in $\expectation{m_{00}}$, increasing in $\expectation{m_{10}}$).
	%
	%
\end{proofof} \\

\begin{proofof}{Proposition~\ref{prop:baseline-util}}
	Let $\mathcal{U}_i(\emptyset)$ denote the expected utility of agent $i$ under her myopic strategy in the no-information setting (that is, randomly matching with another agent).
	\begin{align*}
	\mathcal{U}_i(\emptyset) &= \EEc{u}{\theta_i = 1}\PP\left[\theta_i = 1\right]+\EEc{u}{\theta_i = 0}\PP\left[\theta_i = 0\right]\\
	&= p\Big(pu(1,1)+(1-p)u(1,0)\Big)+(1-p)\Big(pu(1,0)+(1-p)u(0,0)\Big)\\
	&= p\Big(u(1,1)-(1-p)\left(u(1,1)-u(1,0)\right)\Big)+(1-p)\Big(u(0,0)+p\left(u(1,0)-u(0,0)\right)\Big)\\
	&= pu(1,1)+(1-p)u(0,0)+p(1-p)\Big(2u(1,0)-u(1,1)-u(0,0)\Big).
	\end{align*}\end{proofof}

\end{document}